\documentclass[lettersize,journal]{IEEEtran}
\usepackage{amsmath,amsfonts,amsthm}
\usepackage{algorithmic}
\usepackage{algorithm}
\usepackage{array}
\usepackage{subfigure}
\usepackage[caption=false,font=normalsize,labelfont=sf,textfont=sf]{subfig}
\usepackage{textcomp}
\usepackage{stfloats}
\usepackage{url}
\usepackage{verbatim}
\usepackage{graphicx}
\usepackage{cite}
\usepackage{dsfont}
\usepackage{xcolor}
\usepackage{multirow}
\usepackage{flushend}
\usepackage{soul}

\newtheorem{thm1}{\bf Theorem}
\newtheorem{prop1}{\bf Proposition}
\newtheorem{lem1}{\bf Lemma}
\newtheorem{assmpt1}{\bf Assumption}
\newtheorem{defn17}{\bf Definition}
\newtheorem{rem1}{\bf Remark}

\newtheorem{cor1}{\bf Corollary}

\newenvironment{definition}{\begin{defn17}}{\hfill$\Diamond$\end{defn17}}
\newenvironment{remark}{\begin{rem1}}{\hfill$\Diamond$\end{rem1}}

\newenvironment{theorem}{\begin{thm1}}{\hfill$\Diamond$\end{thm1}}

\begin{document}


\title{K-SMPC: Koopman Operator-Based Stochastic Model Predictive Control for Enhanced Lateral Control of Autonomous Vehicles
}

\author{Jin Sung Kim,~\IEEEmembership{Graduate Student Member,~IEEE}, Ying Shuai Quan,~\IEEEmembership{Graduate Student Member,~IEEE,} \\
and Chung~Choo~Chung,~\IEEEmembership{Member,~IEEE}

\thanks{Manuscript received XXX, XX, 2023; revised XXX, XX, 2023; accepted XXX, XX, 2023. This work was supported by the National Research Foundation of Korea (NRF) grant funded by the Ministry of Science and ICT (MSIT) (No.2021R1A2C2009908, Data-Driven Optimized Autonomous Driving Technology Using Open Set Classification Method). \textit{(Corresponding author: Chung Choo Chung.)}}
\thanks{Jin Sung Kim, Ying Shuai Quan are with the Dept. of Electrical Engineering, Hanyang University, Seoul 04763, Korea.
(e-mail: {\tt \{jskim06, ysquan\}@hanyang.ac.kr})}
\thanks{Chung Choo Chung is with the Div. of Electrical and Biomedical Engineering, Hanyang University, Seoul 04763, Korea. (+82-2-2220-1724, e-mail: {\tt cchung@hanyang.ac.kr})}
}

\markboth{IEEE TRANSACTIONS ON INTELLIGENT VEHICLES, VOL. xx, NO. xx, xxx 2023}%
{Shell \MakeLowercase{\textit{et al.}}: A Sample Article Using IEEEtran.cls for IEEE Journals}


\maketitle

\begin{abstract}
This paper proposes Koopman operator-based Stochastic Model Predictive Control (K-SMPC) for enhanced lateral control of autonomous vehicles. The Koopman operator is a linear map representing the nonlinear dynamics in an infinite-dimensional space.
Thus, we use the Koopman operator to represent the nonlinear dynamics of a vehicle in dynamic lane-keeping situations.
The Extended Dynamic Mode Decomposition (EDMD) method is adopted to approximate the Koopman operator in a finite-dimensional space for practical implementation.
We consider the modeling error of the approximated Koopman operator in the EDMD method. Then, we design K-SMPC to tackle the Koopman modeling error, where the error is handled as a probabilistic signal.
The recursive feasibility of the proposed method is investigated with an explicit first-step state constraint by computing the robust control invariant set.
A high-fidelity vehicle simulator, i.e., CarSim, is used to validate the proposed method with a comparative study.
From the results, it is confirmed that the proposed method outperforms other methods in tracking performance. Furthermore, it is observed that the proposed method satisfies the given constraints and is recursively feasible.
\end{abstract}


\section*{NOMENCLATURE}
\begin{itemize}
\item \{${XYZ}$\} : Global coordinate frame
\item \{${xyz}$\} : Local coordinate frame
\item $C_{\alpha i}$ : Cornering stiffness of tire, $i \in \{f, r\}$
\item $V_x$ : Longitudinal speed
\item $V_y$ : Lateral speed
\item $m$ : Total mass of vehicle
\item $l_i$ : Distance between front (rear) tire and center of gravity (CG) , $i \in \{f, r\}$
\item $I_z$ : Yaw moment of inertia of vehicle
\item $a_y$ : Lateral acceleration in $\{xyz\}$
\item $L$ : Look-ahead distance
\item $e_y=y-y_{des}$ : Lateral position error w.r.t. lane
\item $e_{yL}$ : Lateral position error on look-ahead point w.r.t. lane
\item $\psi$ : Yaw angle of vehicle in global coordinate
\item $e_{\psi}=\psi_{des}-\psi$ : Heading angle error in local coordinate w.r.t. lane
\item $\delta$ : Steering angle
\item $R$: Turning radius
\item $\beta$: Vehicle side slip angle at CG
\end{itemize}

\section{Introduction}

Autonomous driving vehicles provide advanced driver assistance functions to relieve humans from monotonous long drives and can significantly decrease traffic congestion and accidents. A typical autonomous driving setup comprises essential components such as perception, localization, decision-making, trajectory planning, and control. During trajectory planning and control, knowledge of vehicle dynamics is necessary to execute accurate and safe maneuvers, particularly in complex and unpredictable road environments.
Thus, it is essential to have a lateral vehicle dynamic model to design a lateral controller. Lateral control of autonomous driving has gained much attention in many areas, such as automated parking control~\cite{seo2021lpv}, lateral control on curved roads~\cite{quan2022linear, choi2021horizon}, and automated lane change systems~\cite{kim2020lane}.
The bicycle lateral dynamic motion model has been widely used to develop lateral control~\cite{rajamani2011vehicle}. In the dynamic model, lateral tire force and acceleration are used to capture the dynamic motion of a vehicle for high-speed driving to represent accurate vehicle behavior.
Although many studies have used the bicycle lateral dynamic model, i.e., linear dynamic model, for practical applications under certain conditions, such as a small tire slip angle with a given vehicle speed, the nonlinearity of the vehicle dynamics cannot be ignored because the tire model is highly nonlinear due to the vertical load transfer~\cite{kim2019vehicular}. Moreover, the vehicle speed is no longer constant in dynamic driving.
Therefore, obtaining a model that captures the full vehicle dynamics for various driving conditions is necessary. 

Numerous studies have attempted to identify the unknown nonlinear dynamics in different research fields~\cite{brunton2022data}. Recently, a modeling approach based on many datasets has received significant attention for complex systems whose dynamics are difficult to capture~\cite{brunton2016discovering}.
In this context, the Koopman operator has been used in model identification of complex dynamics in recent years. The Koopman operator is a linear map for representing nonlinear systems in an infinite-dimensional space~\cite{koopman1931hamiltonian}.
One of the primary benefits of using the Koopman operator is that the linear model can express the underlying nonlinear behavior. As a result, a linear control design method can be applied to a general nonlinear dynamic system.

In recent years, the Koopman operator-based modeling and control approach has been widely adopted in automated driving because vehicles have highly nonlinear behaviors.
In~\cite{cibulka2019data, xiao2022deep, vsvec2021model, cibulka2020model, vsvec2021predictive}, the authors proposed model identification of nonlinear vehicle dynamics to control vehicle lateral and/or longitudinal velocity.
In~\cite{yu2022autonomous} and~\cite{xiao2023ddk}, the authors considered the global position control of the vehicle. Position control is essential for controlling vehicles properly on roads.
In~\cite{wang2021deep}, the authors considered the local position with respect to the given trajectory. Then, a mini-sized car was used to show the effectiveness of the proposed system.
For practical implementation of the Koopman operator, the papers mentioned above used Extended Dynamic Mode Decomposition (EDMD) or neural networks to approximate the Koopman operator in a finite-dimensional space.
Unfortunately, the approximated Koopman operator causes approximation uncertainty, which results in the presence of modeling errors because there is a residual term in the optimization problem of approximation of the Koopman operator~\cite{korda2018linear, zhang2022robust, kim2023uncertainty}.
Therefore, the model mismatch can not be negligible in using the Koopman operator, even though the Koopman operator has a powerful linear property representing the nonlinear dynamics.
To tackle this problem, \cite{son2020handling} proposes a method of handling the approximation error with an estimator.
In~\cite{zhang2022robust} and~\cite{wang2022data}, the authors design Robust Model Predictive Control (RMPC) for the nonlinear system with constraints satisfaction under uncertainties. However, it is challenging to obtain a robust positively invariant set against the uncertainties of the approximated Koopman operator. Moreover, an RMPC scheme is conservative since the worst-case of uncertainties is considered~\cite{farina2016stochastic, mesbah2016stochastic}.

In this context, this paper proposes a Koopman operator-based stochastic MPC (K-SMPC) for enhanced lateral control of autonomous vehicles. The EDMD method is adopted to obtain the approximated Koopman operator in a finite-dimensional space for practical use of the Koopman operator.
Our work considers the modeling error due to the EDMD-based finite-dimensional approximation of the Koopman operator. We handle the modeling error as a probabilistic signal and design the SMPC to tackle the error. As a result, the proposed method is less conservative than the RMPC. To our knowledge, this paper is the first research in which the SMPC is used to resolve the modeling error of the approximated Koopman operator in the LKS application.
All constraints are satisfied under the Koopman modeling error with recursive feasibility in the proposed method.
A high-fidelity vehicle simulator, CarSim, is used to validate the proposed method. The simulation results confirmed that the proposed method always satisfies the constraints and is recursively feasible.
Moreover, a comparative study shows that the proposed method outperforms other methods: the linear vehicle model-based SMPC (L-SMPC) and the Koopman-based Linear Quadratic Regulator (K-LQ)~\cite{kim2023koopman}.
The contributions of the paper are three-fold:
\begin{itemize}
\item
We compute the Koopman-based vehicle model for the Lane Keeping System (LKS). The vehicle model has highly nonlinear dynamic motion in dynamic driving, such as varying vehicle speed or cornering stiffness. Thus, we reformulate the Koopman-based vehicle model from~\cite{kim2023koopman} to effectively capture the vehicle nonlinear dynamics for the LKS in various driving situations.

\item
Uncertainty of the approximated Koopman operator is considered and handled by the probabilistic constraints of K-SMPC. Since the approximated Koopman model may fail to represent the system accurately, we designed K-SMPC to predict the expected state of the system. With the proposed algorithm, we generated K-SMPC resistant to error in the model identification and randomness in the dynamics.

\item
We prove the recursive feasibility of the proposed K-SMPC with an explicit first-step state constraint by computing a robust control invariant set by providing a theorem. Compared to a mixed worst-case/stochastic prediction for constraint tightening, the proposed method is less conservative but has recursive feasibility.
\end{itemize}

The rest of the paper is structured as follows: Section~\ref{sec:vehicle dynamics} investigates the vehicle nonlinear dynamics. Section~\ref{sec:koopman} introduces the background of the Koopman operator theory and its application to vehicle dynamics for the LKS. Based on the obtained Koopman operator, Section~\ref{sec:smpc} presents the SMPC design process with recursive feasibility. The simulation results are shown in Section~\ref{sec:results}, and the conclusion of the paper is described in Section~\ref{sec:conclusion}.


\section{Nonlinear Vehicle Dynamics on Roads}
\label{sec:vehicle dynamics}

\begin{figure}[t]
\includegraphics[width=\columnwidth]{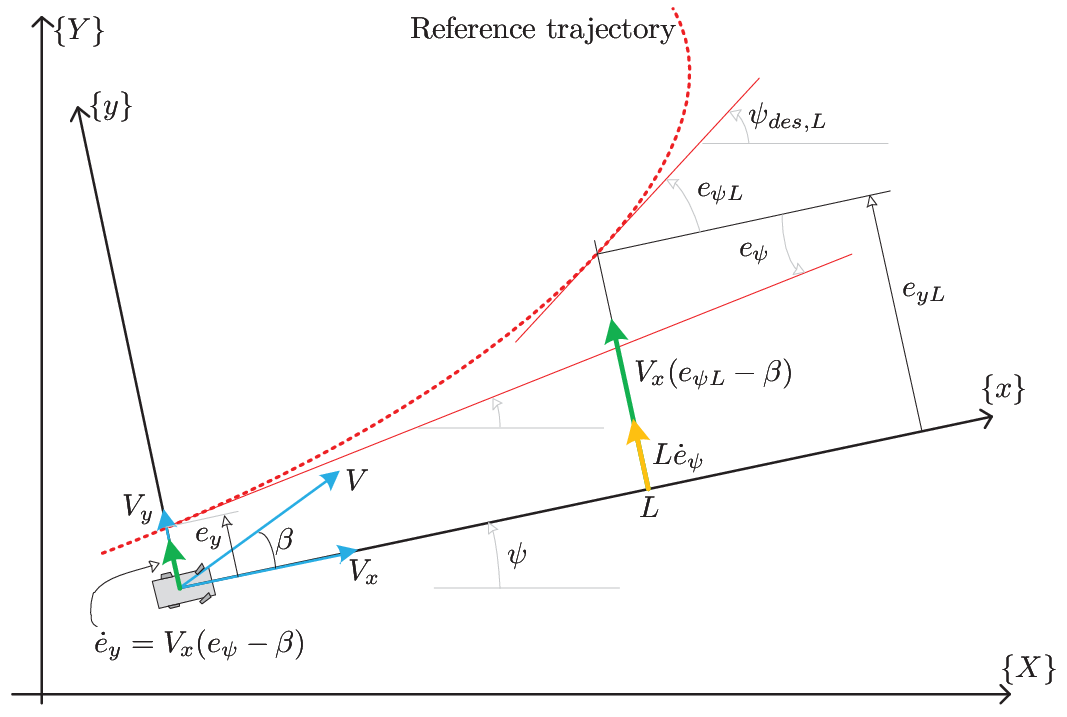}
\caption{Look-ahead lateral dynamic model~\cite{lee2020autonomous}}
\label{fig:look_ahead}
\end{figure}
\subsection{Clothoid Road Lane Model}
\label{subsec:road lane}
We introduce a road lane where a vehicle may run to be represented by a cubic polynomial curve. The cubic polynomial curve is defined by the clothoid curve, where the curvature of the curve is continuous and slowly varying~\cite{lee2020autonomous, choi2021horizon}. To consider the clothoid constraint with slowly varying curvature $\kappa$, it can be defined as
\begin{equation}
\kappa(s) = 2 C_2  + 6 C_3 s,
\label{eq:curvature}
\end{equation}
where $s$ denotes the arc length, $2 C_2$ denotes the curvature at $s=0$, and $6 C_3$ denotes the curvature rate. For a small curvature, the arc length $s$ can be approximated by the longitudinal distance $x$~\cite{rajamani2011vehicle}. Then, integrating~\eqref{eq:curvature} twice leads to a clothoid cubic polynomial road model such that
\begin{equation}
f(x) = C_0 + C_1 x + C_2 x^2 + C_3 x^3,
\label{eq:road model}
\end{equation}
where $C_0$ denotes the lateral offset, and $C_1$ denotes the heading angle error. In general, the lane curve model is widely applied with the assumption of a plain road in a camera-based lane recognition (see~\cite{dickmanns1987curvature} and references therein). It is well known that the road model is obtained by a camera sensor. Moreover, from~\eqref{eq:curvature}, $C_2$ and $C_3$ are the shape of the road, which is not dependent on the vehicle motion. On the other hand, $C_0$ and $C_1$ in~\eqref{eq:road model} are dependent on the vehicle motion since they show the relationship between the vehicle and the road lane curve.

This paper considers the look-ahead distance to mimic human driving behavior~\cite{choi2020robust}. By using~\eqref{eq:road model}, the heading angle error and a lateral offset at look-ahead distance $L$ can be computed as
\begin{equation}
\begin{split}
f(L) &= e_{yL}  = C_0 + C_1 L + C_2 L^2 + C_3 L^3,\\
f'(L) &= e_{\psi L}  = C_1 + 2 C_2 L + 3 C_3 L^2,
\end{split}
\end{equation}
as shown in Fig.~\ref{fig:look_ahead}.


\begin{figure}[t]
\includegraphics[width=\columnwidth]{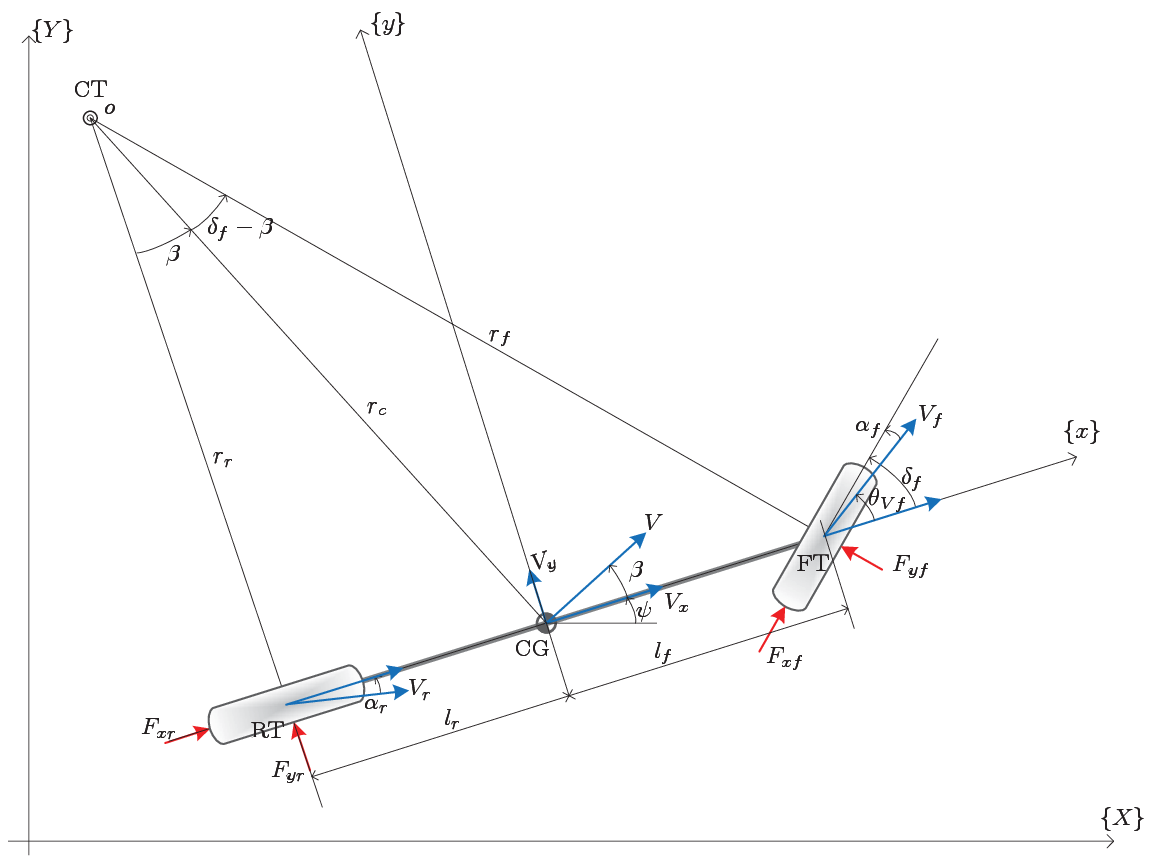}
\caption{Single track bicycle model}
\label{fig:single_track}
\end{figure}
\subsection{Lateral Vehicle Motion Model}

In this subsection, we derive the lateral vehicle motion model as the nonlinear dynamics. To begin with, consider Newton's second law in the lateral direction of the vehicle such that
\begin{equation}
m a_y = F_{yf} + F_{yr}
\label{eq:lateral force}
\end{equation}
where $a_y$ is the lateral acceleration of the vehicle at the
center of gravity and $F_{yf}$ and $F_{yr}$ are the lateral tire forces at the front and rear wheels, respectively. The lateral tire force can be represented as a nonlinear function with respect to the tire slip angle $\alpha_f$, $\alpha_r$, and the vehicle state, which is given by
\begin{equation}
\begin{split}
F_{yf}    &= 2 C_{\alpha f}( \alpha_f) \cdot \Big(  \delta - \arctan(\frac{V_y + l_f \dot{\psi}}{V_x})  \Big), \\
F_{yr}    &= 2 C_{\alpha r}( \alpha_r) \cdot \Big( - \arctan(\frac{V_y - l_r \dot{\psi}}{V_x})  \Big),
\end{split}
\label{eq:tire force}
\end{equation}
where $C_{\alpha f}$ is the cornering stiffness which is a function of the tire slip angle.
There are two terms contributing to the lateral acceleration: the translational acceleration $\ddot{y}$, and the centripetal acceleration $V_x \dot{\psi}$ such that
\begin{equation}
a_y         = \ddot{y} + V_x \dot{\psi}.
\label{eq:lateral accel}
\end{equation}
Substituting~\eqref{eq:lateral force} into~\eqref{eq:lateral accel} leads to
\begin{equation}
\ddot{y}    = -V_x \dot{\psi} + \frac{F_{yf}+F_{yr}}{m}.
\label{eq:ddot y}
\end{equation}
In addition, the yaw dynamics of the vehicle along the z-axis are represented by
\begin{equation}
I_z \ddot{\psi} = l_f F_{yf} - l_r F_{yr},
\label{eq:yaw dynamics}
\end{equation}
where $l_f$ and $l_r$ are the distances of the front wheel and the rear wheel from the center of gravity, respectively.

Let us obtain the heading angle error rate
\begin{equation}
\dot{e}_{\psi}  = \dot{\psi}_{des} - \dot{\psi},
\label{eq:epsi dynamics}
\end{equation}
and the lateral position error rate
\begin{equation}
\dot{e}_{y} = \dot{y} - \dot{y}_{des} = \dot{y} + V_x e_{\psi}.
\label{eq:ey dynamics}
\end{equation}
Then, we can obtain
\begin{equation}
\begin{split}
\ddot{e}_y &= \ddot{y} - \ddot{y}_{des} = \ddot{y} + V_x \dot{e}_{\psi}\\
           &= -V_x \dot{\psi} + \frac{F_{yf}+F_{yr}}{m} + V_x \dot{e}_{\psi}.
\end{split}
\end{equation}
In order to mimic the general behavior of expert drivers, it is necessary to consider error at the look-ahead distance~\cite{lee2016robust}, as shown in Fig.~\ref{fig:look_ahead}. Then, the lateral offset error at the look-ahead distance is given by
\begin{equation}
\begin{split}
\dot{e}_{yL} &= V_x (e_{\psi L} - \beta ) + L \dot{e}_{\psi} \\
             &= \dot{e}_y - L \dot{\psi} + V_x (e_{\psi L} - e_{\psi} ) + L \dot{\psi}_{des}.
\end{split}
\label{eq:eyL dynamics}
\end{equation}
Now, let us define the state, the input, and the external signal of the vehicle dynamics~\cite{lee2016robust, kim2023koopman, kim2022data}
\begin{equation}
\begin{split}
\textbf{x} &=
\begin{bmatrix} e_{y} & e_{yL} & \dot{e}_y & e_{\psi} & \dot{\psi} & a_y & V_y
\end{bmatrix}^T, \\
\textbf{u} &= \delta, \\
\varphi &=
\begin{bmatrix}
 V_x & C_2 & C_3
\end{bmatrix}^T,
\end{split}
\end{equation}
where $\textbf{x} \in \mathds{R}^n$, $\textbf{u} \in \mathds{R}^m$, and $\varphi \in \mathds{R}^d$.
Then, we can describe the nonlinear vehicle dynamics such that
\begin{equation}
\dot{\textbf{x}} = f_v (\textbf{x}, \textbf{u}, \varphi).
\label{eq:continuous nonlinear vehicle model}
\end{equation}
Since the lateral tire force~\eqref{eq:tire force} is highly nonlinear with respect to the tire slip angle and vehicle motion, \eqref{eq:continuous nonlinear vehicle model} can be represented as a nonlinear structure.
%
%
%
Then, discretizing \eqref{eq:continuous nonlinear vehicle model} leads to a discrete-time vehicle nonlinear model given by
\begin{equation}
\textbf{x}_{k+1} = f_d ( \textbf{x}_k, \textbf{u}_k, \varphi_k  ).
\label{eq:discrete nonlinear vehicle model}
\end{equation}
As reported in~\cite{rajamani2011vehicle}, the vehicle dynamics have strong couplings in lateral and longitudinal directions due to the tire characteristics. Thus, it can be challenging to identify the cornering stiffness parameters. Here, this paper tackles this nonlinearity of the vehicle dynamics by taking advantage of an emerging technique in the field of data-driven modeling, i.e., the Koopman operator theory.
It is not necessary to have any prior knowledge of the internal parameters of the vehicle. Only the collected dataset of the system state and input are required to obtain the Koopman operator. Using the property of the Koopman operator, we construct a linear vehicle dynamic model precisely representing~\eqref{eq:discrete nonlinear vehicle model} in a lifted space. We will discuss the detailed design process in the following sections.


\section{Koopman Operator}
\label{sec:koopman}

\subsection{Preliminary}
The Koopman operator was initially proposed to capture the nonlinear autonomous dynamics in an infinite-dimensional space~\cite{koopman1931hamiltonian}. Thus, let us consider the discrete-time nonlinear autonomous dynamics such that
\begin{equation}
\boldsymbol{\eta}_{k+1} = f_a ( \boldsymbol{\eta}_k),
\label{eq:autonomous nonlinear model}
\end{equation}
where $\boldsymbol{\eta}_k \in \mathcal{N}$ is the state of the system, $f_a$ is the nonlinear map, and $k \in \mathds{Z}_+$ is the discrete-time.
Let us consider a real-valued scalar function $\pi_a: \mathcal{N} \rightarrow \mathds{R}$, which is the so-called \emph{observable}~\cite{mauroy2020koopman, brunton2022data}. Each function $\pi_a$ is an element of an infinite-dimensional function space $\mathcal{F}_a$ (i.e., $\pi_a \in \mathcal{F}_a$)~\cite{mauroy2020koopman}. Then, the Koopman theory provides an alternative representation of~\eqref{eq:autonomous nonlinear model} by a linear operator, i.e., the Koopman operator $\mathcal{K}_a: \mathcal{F}_a \rightarrow  \mathcal{F}_a$ in the space $\mathcal{F}_a$ such that
\begin{equation}
\mathcal{K}_a  \pi_a (\boldsymbol{\eta}_k ) = \pi_a (f_a (\boldsymbol{\eta}_k) )
\end{equation}
for every $\pi_a \in \mathcal{F}_a$, where the function space $\mathcal{F}_a$ is invariant under the Koopman operator~\cite{mauroy2020koopman, korda2018linear}.

There are several ways to apply the Koopman operator to controlled nonlinear systems with a slight change~\cite{proctor2018generalizing, williams2016extending, korda2018linear}. This paper adopts the data-driven method from~\cite{korda2018linear}, which is a rigorous and practical approach.
Let us consider a controlled discrete-time nonlinear system such that
\begin{equation}
\boldsymbol{\eta}_{k+1} = f (\boldsymbol{\eta}_{k}, \boldsymbol{\nu}_k),
\label{eq:controlled nonlinear model}
\end{equation}
where $\boldsymbol{\nu}_k \in \mathcal{V}$ is the system input. We can then define the extended state-space $\mathcal{N} \times \mathcal{I} (\mathcal{V})$, where $\mathcal{I} (\mathcal{V})$ is the space of all the control sequences, $\boldsymbol{\mu} := (\boldsymbol{\nu}_k)_{k=0}^{\infty}$. Using the scheme from~\cite{korda2018linear}, we can define the extended state given by
\begin{equation}
\chi = \begin{bmatrix} \boldsymbol{\eta} \\ \boldsymbol{\mu} \end{bmatrix}.
\label{eq:extended state}
\end{equation}
With the extended state~\eqref{eq:extended state},  \eqref{eq:controlled nonlinear model} can be in the form of an autonomous system such that
\begin{equation}
\chi_{k+1} = F(\chi_k) :=
\begin{bmatrix}
f(\boldsymbol{\eta}_k, \boldsymbol{\mu}_k (0) ) \\ \mathcal{L} \boldsymbol{\mu}_k
\end{bmatrix},
\label{eq:extended state model}
\end{equation}
where $\mathcal{L}$  is the left shift operator, i.e., $\mathcal{L} \boldsymbol{\mu}_k = \boldsymbol{\mu}_{k+1}$, and $\boldsymbol{\mu}_k (0)= \boldsymbol{\nu}_k$ is the first element of the control sequence of $\boldsymbol{\mu}$ at the time step $k$~\cite{korda2018linear}.
Now, we can define the Koopman operator $\mathcal{K}_f: \mathcal{F} \rightarrow  \mathcal{F}$ for~\eqref{eq:extended state model} as
\begin{equation}
\mathcal{K}_f \pi (\chi_k ) = \pi (F(\chi_k) ),
\label{eq:koopman definition}
\end{equation}
where $\pi: \mathcal{N} \times \mathcal{I} (\mathcal{V}) \rightarrow \mathds{R}$ is a real-valued function, which belongs to the extended function space $\mathcal{F}$~\cite{mauroy2020koopman}.
%
%
%
%
Interestingly, it is observed that the Koopman operator is linear in the function space $\mathcal{F}$, even though the dynamical system is nonlinear~\cite{mauroy2020koopman}.
%

\subsection{Koopman Operator-based Vehicle Modeling}

In this subsection, we introduce the Koopman operator-based vehicle modeling approach. In~\eqref{eq:koopman definition}, we can see that the Koopman operator ${\mathcal K}_f$ lies in the infinite-dimensional space for representing the original nonlinear dynamics~\cite{koopman1931hamiltonian, mauroy2020koopman}. Thus, it is challenging to directly use the Koopman operator if the finite-dimensional approximation of the Koopman operator is not obtained.
To resolve this problem, this paper uses the EDMD method from~\cite{korda2018linear, williams2015data}. Let us first recall the state, the control input, and the external signal of the vehicle dynamics~\eqref{eq:discrete nonlinear vehicle model} such that
\begin{equation}
\begin{split}
\textbf{x} &=
\begin{bmatrix}
e_y & e_{yL} & \dot{e}_y & e_{\psi} & \dot{\psi} & a_y & V_y
\end{bmatrix}^T,
%
\\
\textbf{u} &= \delta,
%
\\
\varphi &=
\begin{bmatrix}
 V_x & C_2 & C_3
\end{bmatrix}^T.
%
%
\end{split}
\label{eq:state}
\end{equation}
\begin{remark}
Since this paper focuses on vehicle modeling for lateral motion control, the longitudinal speed $V_x$ can be the external signal. In addition, the curvature and curvature rate of the road lane, i.e., $C_2$ and $C_3$, is independent of the vehicle state, as mentioned in Subsection~\ref{subsec:road lane}~\cite{lee2016robust, kim2023koopman}. Thus, $C_2$ and $C_3$ can be the external signal. In general, $\varphi$ is available with an in-vehicle sensor and a camera.
\end{remark}
Then, we take and modify the approach from~\cite{korda2018linear, williams2015data} to define the extended state given by
\begin{equation}
\mathcal{X}_k = \begin{bmatrix}
\textbf{x}_k \\
\textbf{u}_k \\
\varphi_k
\end{bmatrix},
\end{equation}
where $\mathcal{X}_k \in {\mathds R}^{n+m+d}$ is the extended state. Then, we can have the discrete-time autonomous system for the extended state such that
\begin{equation}
\mathcal{X}_{k+1} = F (\mathcal{X}_k) :=
\begin{bmatrix}
f_d (\textbf{x}_k, \textbf{u}_k, \varphi_k )\\
\textbf{u}_{k+1} \\
\varphi_{k+1}
\end{bmatrix}.
\end{equation}
The Koopman operator can then be obtained by
\begin{equation}
{\mathcal K} \boldsymbol{\xi}(\mathcal{X}_k) = \boldsymbol{\xi} ( F (\mathcal{X}_k) ),
\label{eq:koopman finite}
\end{equation}
where $\boldsymbol{\xi}(\textbf{x}_k, \textbf{u}_k, \varphi_k)= \begin{bmatrix} \boldsymbol{\phi} (\textbf{x}_k ) & \textbf{u}_k & \textbf{w}_k\end{bmatrix}^T$ is the lifting function. In this case, we consider $\boldsymbol{\phi} ( {\textbf{x}_k} )$ as
\begin{equation}
\boldsymbol{\phi} ( {\textbf{x}_k} )
= \begin{bmatrix}
\phi_1 ( \textbf{x}_k ) \\
\phi_2 ( \textbf{x}_k ) \\
\vdots \\
\phi_N ( \textbf{x}_k )
\end{bmatrix} \in \mathds{R}^N,
\label{eq:lifting function}
\end{equation}
where $\phi_i: {\mathds R}^n \rightarrow {\mathds R}$ is the real-valued lifting function, and $N \gg n$. In general, the lifting function $\phi_i$ is a user-defined nonlinear function. In this paper, the EDMD method from~\cite{korda2018linear} is used to approximate the Koopman operator in~\eqref{eq:koopman finite} as a finite-dimensional linear operator. The analytical solution is obtained by
\begin{equation}
\min_{\mathcal K}
\sum_{i=0}^{M-1}
\|
\boldsymbol{\xi} ( \mathcal{X}_{i+1} ) -
{\mathcal K} \boldsymbol{\xi} (\mathcal{X}_i)
\|_2^2,
\label{eq:minimization for koopman}
\end{equation}
where $M$ is the length of a dataset. To solve the optimization problem, we first need to collect a dataset by conducting several numerical simulations. Then the dataset matrices are given as
\begin{equation}
\begin{split}
\textbf{X} &=
\begin{bmatrix}
\textbf{x}_0 & \textbf{x}_1 & \dots & \textbf{x}_{M-1}
\end{bmatrix} \in {\mathds R}^{n \times M},
\\
\textbf{U} &=
\begin{bmatrix}
\textbf{u}_0 & \textbf{u}_1 & \dots & \textbf{u}_{M-1}
\end{bmatrix} \in {\mathds R}^{m \times M},
\\
\textbf{D} &=
\begin{bmatrix}
\varphi_0 & \varphi_1 & \dots & \varphi_{M-1}
\end{bmatrix} \in {\mathds R}^{d \times M},
\\
\textbf{Y} &=
\begin{bmatrix}
\textbf{x}_1 & \textbf{x}_2 & \dots & \textbf{x}_M
\end{bmatrix} \in {\mathds R}^{n \times M},
\end{split}
\label{eq:data collection}
\end{equation}
where M is the length of a dataset. Let us define the basis function~$\phi_i$ such that
\begin{equation}
\textbf{z}_k = \boldsymbol{\phi} ( {\textbf{x}_k} ) :=
\begin{bmatrix}
\textbf{x}_k \\
\phi_{N-n}(\textbf{x}_k) \\
\vdots \\
\phi_{N}(\textbf{x}_k)
\end{bmatrix} \in {\mathds R}^{N}.
\label{eq:lifted state}
\end{equation}
Since predicting the future input and external signal is not of interest~\cite{korda2018linear, kim2023uncertainty}, this paper omits the last $(m+d)$ rows of each $\boldsymbol{\xi} ( \mathcal{X}_{i+1} ) - {\mathcal K} \boldsymbol{\xi} (\mathcal{X}_i)$ in~\eqref{eq:minimization for koopman}. However, we focus on the first $N$ rows such that
\begin{equation}
\min_{\mathcal K}
\sum_{i=0}^{M-1}
\Big{\|}
\begin{bmatrix} \boldsymbol{\phi} (\textbf{x}_{k+1} ) \\ \textbf{u}_{k+1} \\ \varphi_{k+1} \end{bmatrix}
-
{\mathcal K}  \begin{bmatrix} \boldsymbol{\phi} (\textbf{x}_k ) \\ \textbf{u}_k \\ \varphi_k \end{bmatrix}
\Big{\|}_2^2
\end{equation}
where
\begin{equation*}
{\mathcal K} = \begin{bmatrix}
A & B & B_{\varphi} \\
(*) & (*) & (*) \\
(*) & (*) & (*)
\end{bmatrix}.
\end{equation*}
Then, \eqref{eq:minimization for koopman} can be converted into
\begin{equation}
\min_{A, B, B_{\varphi}}
\|
\tilde{\textbf{Y}}
-A \tilde{\textbf{X}} - B \textbf{U} - B_{\varphi} \textbf{D}
\|_F^2,
\label{eq:minimization for AB}
\end{equation}
where
\begin{equation*}
\begin{split}
\tilde{\textbf{X}} &=
\begin{bmatrix}
\boldsymbol{\phi} ( \textbf{x}_0 ) &
\boldsymbol{\phi} ( \textbf{x}_1 ) &
\dots &
\boldsymbol{\phi} ( \textbf{x}_{M-1})
\end{bmatrix},
\\
\tilde{\textbf{Y}} &=
\begin{bmatrix}
\boldsymbol{\phi} ( \textbf{x}_1 ) &
\boldsymbol{\phi} ( \textbf{x}_2 ) &
\dots &
\boldsymbol{\phi} ( \textbf{x}_M )
\end{bmatrix},
\end{split}
\end{equation*}
and $\| \cdot \|_F$ is the Frobenius norm. By solving the optimization problem~\eqref{eq:minimization for AB}~\cite{kim2023koopman, korda2018linear}, we can obtain the linear model such that
\begin{equation}
\begin{split}
\textbf{z}_{k+1} &= A \textbf{z}_k + B \textbf{u}_k + B_{\varphi} \varphi_k + G \textbf{w}_k, \\
\textbf{x}_{k} &= C \textbf{z}_k.
\end{split}
\label{eq:lift_sys}
\end{equation}
where the reconstruction matrix $C$ is obtained by $C=\begin{bmatrix} \boldsymbol{I}_{(n \times n)} & \textbf{0} \end{bmatrix}$. Here, note that this paper introduces the residual term $\textbf{w}_k$ in~\eqref{eq:lift_sys}.
This is because there may be a residual term in solving~\eqref{eq:minimization for AB}, which results in the approximation error of the Koopman operator~\cite{korda2018linear, zhang2022robust, kim2023uncertainty}.
Thus, the modeling error of the tuplet $(A, B, B_{\varphi})$ is inevitable due to the approximated Koopman operator in a finite-dimensional space.
To resolve the problem, we consider the residual term~$\textbf{w}_k$ of the Koopman-based model in designing the controller. Moreover, we assume that $\textbf{w}_k$ is the bounded probabilistic signal such that $\textbf{w}_k \in \mathcal{W}=\{ \textbf{w}_k | \|\textbf{w}_k\|_{\infty} \leq \bar{\textbf{w}} \}$, $\mathds{E}[\textbf{w}_k]=\textbf{0}$, and the covariance matrix of $\textbf{w}_k$ is $ \Sigma_{\textbf{w}}$. In the following subsection, we will introduce the design process of the proposed controller considering the approximation error $\textbf{w}_k$.

\begin{figure}[t]
\center
\includegraphics[width=\columnwidth]{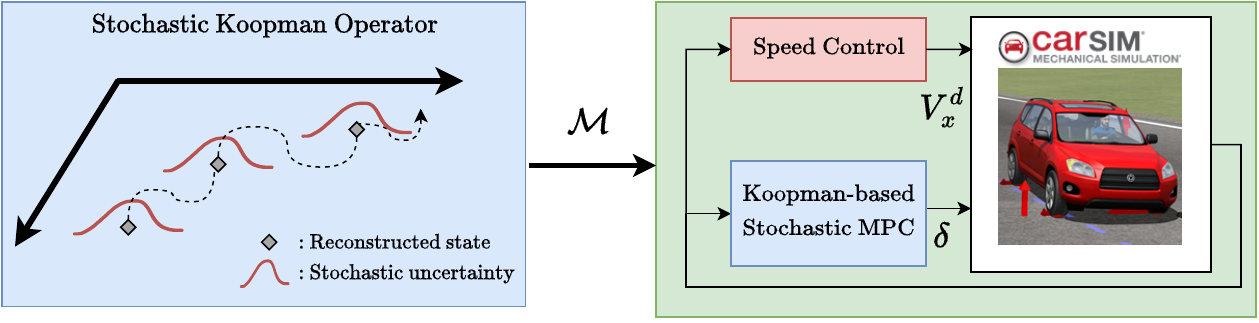}
\caption{Schematic illustration of the proposed method. The approximation error of the Koopman operator is handled as a stochastic uncertainty.}
\end{figure}
\section{Koopman Operator-based Stochastic Model Predictive Control}
\label{sec:smpc}

\subsection{System State, Objective, and Constraints}

In this subsection, we first describe the system state to be controlled. One can denote the predicted trajectories with $k+i|k$, i.e., predicted at time $k$ and $i$ steps into the future. We define $\textbf{z}_{k+i|k}$ as
\begin{equation}
\textbf{z}_{k+i|k} = \textbf{s}_{k+i|k} + \textbf{e}_{k+i|k},
\label{eq:state decomposition}
\end{equation}
where the state $\textbf{z}_{k+i|k}$ is decomposed into two parts: the deterministic state $\textbf{s}_{k+i|k}$ and the zero mean stochastic error $\textbf{e}_{k+i|k}$, i.e., $\mathds{E}[\textbf{z}_{k+i|k}]= \textbf{s}_{k+i|k}$. Let us define the stabilizing control gain $K$ satisfying the following Riccati equation such that
\begin{equation}
P = A^T P A - A^T P B ( R + B^T P B)^{-1} B^T P A + Q
\label{eq:riccati equation}
\end{equation}
where $K = (R + B^T P B)^{-1} B^T P A$. Then, as it is common in the linear SMPC scheme, e.g.,~\cite{farina2016stochastic}, the control strategy is given by
\begin{equation}
\textbf{u}_{k+i|k} = K \textbf{z}_{k+i|k} + \textbf{v}_{k+i|k}
\label{eq:control input}
\end{equation}
where $\textbf{v}_{k+i|k} \in \mathds{R}^m$ is the optimal control input obtained by solving the SMPC problem. Using \eqref{eq:state decomposition} and \eqref{eq:control input}, one can derive the dynamics of the deterministic state and error state given by
\begin{subequations}
\begin{align}
\textbf{s}_{k+i+1|k} &= A_{cl} \textbf{s}_{k+i|k} + B \textbf{v}_{k+i|k} +B_{\varphi} \varphi_{k+i|k} \label{eq:deterministic state dynamics}\\
\textbf{e}_{k+i+1|k} &=  A_{cl} \textbf{e}_{k+i|k} + G \textbf{w}_{k+i|k}
\label{eq:error dynamics}
\end{align}
\end{subequations}
where $A_{cl} = A - BK$ is strictly stable.
\begin{remark}
As mentioned in~\ref{subsec:road lane}, $C_2$ and $C_3$ are intrinsic parameters of a road shape independent of the vehicle's lateral motion~\cite{rajamani2011vehicle}. Thus, with a given road, it is immediate to obtain $C_2$ and $C_3$ in the prediction horizon~\cite{son2014robust}. Moreover, the vehicle speed can be obtained with speed planning and control according to the road curvature~\cite{quan2022linear}. Therefore, this paper assumes that $\varphi$ is available in the horizon $N$.
\end{remark}

Let the cost function in a stochastic framework be
\begin{equation}
\begin{split}
\mathcal{J} &= \mathds{E} \Bigg[ \sum_{i=0}^{N_p-1} \Big( \textbf{z}_{k+i|k}^T Q_{xx} \textbf{z}_{k+i|k} + \textbf{z}_{k+i|k}^T Q_{xv} \textbf{v}_{k+i|k} \\
&~~~~~~~~+ \textbf{v}_{k+i|k}^T Q_{vv} \textbf{v}_{k+i|k} \Big) + \textbf{z}_{N_p|k}^T P \textbf{z}_{N_p|k} \Bigg],
\end{split}
\label{eq:stochastic cost fcn}
\end{equation}
where $\mathds{E} [ \cdot ]$ denotes the expectation value, $Q_{xx} \succeq 0$, $Q_{xv} \succeq 0$, $Q_{vv} \succ 0$, and $P$ is the solution to~\eqref{eq:riccati equation}. Substituting~\eqref{eq:state decomposition} into~\eqref{eq:stochastic cost fcn} leads to the cost function in a deterministic framework by using $\mathds{E}[\textbf{z}_{k+i|k}]= \textbf{s}_{k+i|k}$ such that
\begin{equation}
\begin{split}
\mathcal{J} &=  \sum_{i=0}^{N_p-1} \Big( \textbf{s}_{k+i|k}^T Q_{xx} \textbf{s}_{k+i|k} \\
&~~~~~~~~+ \textbf{s}_{k+i|k}^T Q_{xv} \textbf{v}_{k+i|k} + \textbf{v}_{k+i|k}^T Q_{vv} \textbf{v}_{k+i|k} \Big) \\
&~~~~~~~~+ \textbf{s}_{k+N_p|k}^T P \textbf{s}_{k+N_p|k} + c,
\end{split}
\label{eq:deterministic cost fcn}
\end{equation}
where $c=\mathds{E} [\sum_{i=0}^{N_p-1} ( \textbf{e}_{k+i|k}^T Q_{xx} \textbf{e}_{k+i|k} ) + \textbf{e}_{N_p|k}^T P\textbf{e}_{N_p|k} ]$ which does not depend on the decision variables $\textbf{v}_{k+i|k}$. Thus, we can convert the stochastic cost function into the deterministic cost function.

In terms of the stochastic error and its influence on the deterministic state, the state constraints at the $i$-th time step in the receding horizon can be described as probabilistic constraints based on the risk level or allowable probability of violation, i.e., $\epsilon_i \in [0,1]$, such that
\begin{equation}
\mathcal{P}[H_i \textbf{z}_{k+i|k} \leq h_i] \geq 1 - \epsilon_i,
\label{eq:chance constraint}
\end{equation}
where $\mathcal{P} [ \cdot ] $ denotes the probability, $H_i \in \mathds{R}^{p \times N}$, and $h_i \in \mathds{R}^{p}$. Then, the following theorem provides the process to convert~\eqref{eq:chance constraint} into the deterministic constraints.

\begin{theorem}
At time $k$ with a given prediction horizon $N_p$, the probabilistic constraints of~\eqref{eq:chance constraint} are satisfied if and only if the following deterministic constraints are satisfied such that
\begin{equation}
H_i \textbf{s}_{k+i|k} \leq h_i - q_i(1 - \epsilon_i), \text{ for } i=0, \cdots, N_p-1
\end{equation}
where $q_i(1-\epsilon_i)=\sqrt{H_i^T \Sigma_i H_i} \sqrt{\frac{1-\epsilon_i}{\epsilon_i}}$.
\end{theorem}
\begin{proof}
By using~\eqref{eq:state decomposition}, we can rewrite~\eqref{eq:chance constraint} as
\begin{equation}
\mathcal{P}[H_i \textbf{s}_{k+i|k} \leq h_i - H_i \textbf{e}_{k+i|k}] \geq 1 - \epsilon_i.
\end{equation}
Then, we can obtain
\begin{equation}
H_i \textbf{s}_{k+i|k} \leq h_i  - q_i(1 - \epsilon_i)
\label{eq:constraint tightening}
\end{equation}
where $\mathcal{P}[-q_i(1 - \epsilon_i) \leq -H_i \textbf{e}_{k+i|k}]=1 - \epsilon_i$ because $\textbf{s}_{k+i|k}$ is the deterministic variable. Then, it is immediate to derive $q_i(1-\epsilon_i)=\sqrt{H_i^T \Sigma_i H_i} \sqrt{\frac{1-\epsilon_i}{\epsilon_i}}$ by Chebyshev\textendash Cantelli Inequality~\cite{marshall1960multivariate}, where $\Sigma_{i+1} = A_{cl}^T \Sigma_i A_{cl} + G^T \Sigma_{\textbf{w}} G$ with $\Sigma_{0}=\Sigma_{w}$.
\end{proof}
Consequently, we can define the sets of the deterministic state constraints and input hard constraints for the K-SMPC as
\begin{subequations}
\label{eq:constraint set}
\begin{align}
\mathcal{S} &= \{ \textbf{s}_{k+i|k} \in \mathds{R}^{N} ~|~  H_i \textbf{s}_{k+i|k} \leq h_i - q_i(1 - \epsilon_i) \}, \label{eq:constraint set_state}\\
\mathcal{U} &= \{ \textbf{u}_{k+i|k} \in \mathds{R}^m ~|~ \underline{\textbf{u}} \leq \textbf{u}_{k+i|k} \leq \overline{\textbf{u}}  \}, \label{eq:constraint set_input}
\end{align}
\end{subequations}
where $\underline{\textbf{u}}$, and $\overline{\textbf{u}}$ denote the lower bound, and the upper bound, respectively.
Moreover, a constraint tightening method similar to~\eqref{eq:constraint tightening} can be applied to define the terminal region such that
\begin{equation}
\mathcal{S}_f = \{ \textbf{s}_{k+N_p|k} \in \mathds{R}^{N} ~|~ H_N \textbf{s}_{k+N_p|k} \leq h_{N_p} - q_N (1 - \epsilon_{N_p})\}.
\end{equation}

\begin{figure}[t]
\centering
\includegraphics[width=\columnwidth]{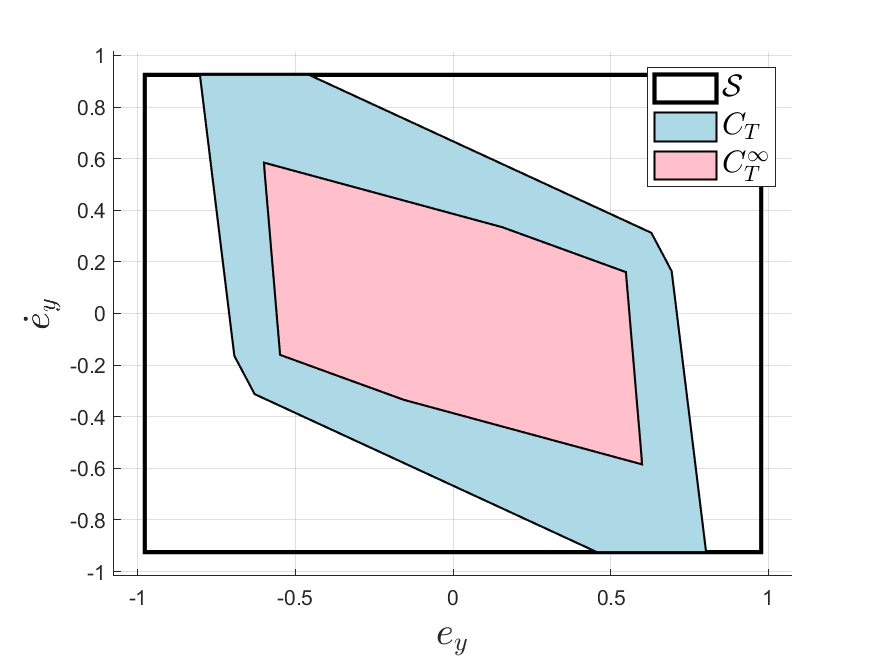}
\caption{Recursive set projected on the space of the first and third state of $\textbf{s}_k$.}
\label{fig:RecursiveSet}
\end{figure}
\subsection{Recursive Feasibility and Stability of Resulting K-SMPC Algorithm}

In order to guarantee the recursive feasibility of the K-SMPC, we construct the first-step state constraint of the prediction horizon~\cite{lorenzen2016constraint}. In~\cite{kouvaritakis2010explicit}, it was reported that the probability of the constraint satisfaction in $i$ steps of the prediction horizon at time $k$ is not equal to the probability of the constraint satisfaction in $i-1$ steps of the prediction horizon at time $k+1$.
Thus, we need to use further constraints to satisfy the recursive feasibility. In~\cite{kouvaritakis2010explicit}, the authors proposed a mixing stochastic and worst-case state prediction in constraint tightening for recursive feasibility in the presence of perturbation. However, in~\cite{korda2011strongly}, the authors point out the mixed stochastic/worst-case approach is rather restrictive and has higher average costs if the solution is near a chance constraint.
Instead,~\cite{korda2011strongly} proposed the constraint only in the first step of the prediction horizon where only the recursive feasibility is of interest. Therefore, we focus on the first step state constraint for recursive feasibility, proposed by a paper in the model-based setting~\cite{lorenzen2016constraint}. Thus, the proposed method is less conservative than the mixed-state prediction approach, e.g.,~\cite{kouvaritakis2010explicit}.

Let us define the following set
\begin{equation*}
  C_T = \left\{
\begin{tabular}{ c | l }
	\multirow{3}{4.3em}
	{$\textbf{s}_{0|k} \in \mathds{R}^{N}$}
&~ $\exists \textbf{v}_{0|k}, \cdots, \textbf{v}_{N_p-1|k}$ \\
&~ \eqref{eq:deterministic state dynamics} and \eqref{eq:constraint set} hold\\
&~ $\textbf{s}_{k+N_p | k} \in \mathcal{S}_f $
\end{tabular}
  \right\}
\end{equation*}
as the $T$-step set with a feasible first step state constraint for the deterministic system~\eqref{eq:deterministic state dynamics} under tightened constraints. The $T$-step set is obtained by the backward recursion from~\cite{gutman1987algorithm}. Since $C_T$  is not necessarily robust positively invariant with respect to the disturbance set $\mathcal{W}$, further computation of the robust control invariant polytope $C_T^{\infty}$ is required. To calculate $C_T^{\infty}$, let us define a set as
\begin{equation*}
  C_T^{i+1}  = \left\{
\begin{tabular}{ c | l }
	\multirow{2}{3em}
	{$\textbf{s}  \in C_T^{i}$}
&~ $\exists v_{0|k} \text{ such that }\eqref{eq:constraint set_input} \text{ holds,}$ \\
&~ $ \textbf{s}_{k+1} \in C_T^{i} \ominus G \mathcal{W}.$
\end{tabular}
  \right\}
\end{equation*}
The set $C_T^{\infty}$ is then computed by $C_T^{\infty}=\cap_{i=0}^{\infty} C_T^i$, where the initial set is $C_T^0=C_T$. The recursive computation method can provides the $C_T^{\infty}$ until $C_T^{i}=C_T^{i+1}$~\cite{blanchini2008set, lorenzen2016constraint}.
This paper adopts the Multi-parametric toolbox from~\cite{herceg2013multi} in MATLAB to compute the set $C_T^{\infty}$, as shown in Fig.~\ref{fig:RecursiveSet}.

In this paper, we additionally consider the soft constraints on the first-step input. Thus, the slack variables, i.e., $\underline{\sigma} \in \mathds{R}$ and $\overline{\sigma}  \in \mathds{R}$, are used in the cost function given by
\begin{equation}
\mathcal{J}_s = \mathcal{J} +  \underline{\sigma}^T S  \underline{\sigma} + \overline{\sigma}^T S  \overline{\sigma}
\end{equation}
where $S > 0$. Then, we have the final K-SMPC algorithm such that
\begin{subequations}
\label{eq:smpc problem}
\begin{align}
&\textbf{v}_{\cdot|k}^{*} = \underset {\textbf{v}_{k+i|k}} {\text{arg min}} ~ \mathcal{J}_s \label{eq:smpc problem_CostFcn}\\
\text{subject}&~\text{to} \nonumber \\
&\textbf{s}_{k+i+1|k} = A_{cl} \textbf{s}_{k+i|k} + B \textbf{v}_{k+i|k} +B_{\varphi} \varphi_{k+i|k} \label{eq:smpc problem_system},\\
&\textbf{s}_{k+i|k} \in \mathcal{S} \label{eq:smpc problem_StateCon},\\
&\textbf{u}_{k+i|k} \in \mathcal{U} \label{eq:smpc problem_InputHardCon},\\
&\textbf{s}_{k+1|k} \in C_T^{\infty} \ominus G \mathcal{W} \label{eq:smpc problem_recursive},\\
%
%
&\underline{\textbf{u}}_{s} - \underline{\sigma} \leq \textbf{u}_{k|k} \leq \overline{\textbf{u}}_{s} + \overline{\sigma} \label{eq:smpc problem_InputSoftCon},\\
& 0 \leq \underline{\sigma} \leq  \underline{\textbf{u}}_{s} - \underline{\textbf{u}},  ~~ 0 \leq  \overline{\sigma} \leq  \overline{\textbf{u}} - \overline{\textbf{u}}_{s} \label{eq:smpc problem_slack},\\
& \textbf{s}_{k+N|k} \in \mathcal{S}_f \label{eq:smpc problem_TerminalCon},\\
&\textbf{s}_{k|k} = \textbf{z}_{k|k},
%
%
\\
& i \in \{0, \dots,N-1 \},
\end{align}
\end{subequations}
where $\underline{\textbf{u}}_{s} \in \mathds{R}$ and $\overline{\textbf{u}}_{s} \in \mathds{R}$ are the upper and lower bound for the first control input, respectively.
In~\eqref{eq:smpc problem_InputSoftCon}, we can see that the soft constraint is used in the first-step input. We also consider the slack variable to satisfy the input constraint~\eqref{eq:smpc problem_InputHardCon} by imposing~\eqref{eq:smpc problem_slack}.
\begin{remark}
We impose the input constraint~\eqref{eq:smpc problem_InputHardCon} to consider the physically bounded front tire angle $\delta$ of vehicles. In addition, it is needed to minimize the tire angle on straight roads or curved roads with small curvature. To that end, we additionally impose the input constraint~\eqref{eq:smpc problem_InputSoftCon} with slack variables, i.e., $\underline{\sigma}$ and $\overline{\sigma}$.
\end{remark}
As mentioned above, the recursive feasibility is guaranteed by the constraint~\eqref{eq:smpc problem_recursive}. Moreover, the following theorem provides the details of the recursive feasibility and its proof.
\begin{theorem}[Recursive Feasibility~\cite{lorenzen2015improved}]
Let us consider the lifted system~\eqref{eq:lift_sys} with the controller~\eqref{eq:control input}. If there exists a feasible solution when $k=0$, then the optimization problem~\eqref{eq:smpc problem} is feasible for $k>0$.
\end{theorem}
\begin{proof}
If the K-SMPC optimization problem~\eqref{eq:smpc problem} is feasible at $k=0$, then $\textbf{s}_{k+1|k} \in C_T^{\infty} \ominus G \mathcal{W}$. In the next time step, we can obtain $\textbf{z}_{k+1}=\textbf{s}_{k+1|k} + G \textbf{w}_k \in C_T^{\infty}$ for every realization $\textbf{w}_k \in \mathcal{W}$, i.e., $\textbf{z}_{k+1}$ is the feasible state in the next time. Therefore, the K-SMPC optimization problem~\eqref{eq:smpc problem} is recursively feasible. Refer to~\cite{lorenzen2015improved} for more details.
\end{proof}

In order to prove the stability of the closed-loop system constructed by~\eqref{eq:smpc problem}, we introduce a discrete-time Input-to-State Stability (ISS) Lyapunov function~\cite{jiang2001input}.
\begin{definition}[ISS-Lyapunov function~\cite{jiang2001input}]
A function $V:\mathds{R}^N \rightarrow \mathds{R}_+$ is an ISS-Lyapunov function for system $\textbf{z}_{k+1}=f_L (\textbf{z}_{k}, \mu_k)$ if the following holds:
\begin{itemize}
\item There exist $\mathcal{K}_{\infty}$ functions $\alpha_1$, $\alpha_2$ such that
    \[\alpha_1(\| \textbf{z} \|) \leq V(\textbf{z})
    \leq \alpha_2(\| \textbf{z} \|), ~~\forall \textbf{z} \in \mathds{R}^N. \]
\item There exist a $\mathcal{K}_{\infty}$ function $\alpha_3$ and a $\mathcal{K}$ function $\gamma$ such that
    \[V(f_L (\textbf{z}, \mu)) - V(\textbf{z}) \leq - \alpha_3(\| \textbf{z} \|) + \gamma (\| \mu \|)\]
    for all $\textbf{z} \in \mathds{R}^N$, and $\mu \in \mathcal{M}$.
\end{itemize}
\label{def:ISS}
\end{definition}
\noindent
Using Definition~\ref{def:ISS}, the following theorem provides the stability of the closed-loop system.
\begin{theorem}[Stability of closed-loop system]
If feasibility of~\eqref{eq:smpc problem} at $k=0$ is given, then the closed-loop system~\eqref{eq:smpc problem} under the proposed controller is input-to-state stable with the ISS-Lyapunov function
\begin{equation*}
V( \textbf{z}_k^* ) = \mathds{E} \Big\{ \sum_{i=0}^{N_p-1}   \big( \| \textbf{z}_{k+i|k}^* \|^2_Q + \| \textbf{u}_{k+i|k}^* \|^2_R \big)
+  \| \textbf{z}_{k+N_p|k}^* \|^2_P\Big\}.
\end{equation*}
\end{theorem}
\begin{proof}
Let $V( \textbf{z}_{k}^* )$ and $V( \textbf{z}_{k+1}^* )$ be an ISS-Lyapunov candidate function at time $k$ and $k+1$, respectively. With the stabilizing control input after prediction horizon $\textbf{u}_{k+N|k}= K \textbf{z}_{k+N|k}$, we have
\begin{equation}
\begin{split}
&\mathds{E}\{ V( \textbf{z}_{k+1}^* )\} - V( \textbf{z}_k^* ) \\
&~ = \mathds{E} \Big\{  \sum_{i=1}^{N_p-1} \big(
\| \textbf{z}_{k+i|k}^* \|^2_Q
+ \| \textbf{u}_{k+i|k}^* \|^2_R
\big)
+ \| \textbf{z}_{k+N_p|k}^* \|^2_Q \\
&~~~~ +  \| \textbf{u}_{k+N_p|k}^* \|^2_R
+ \| \textbf{z}_{k+N_p+1|k}^* \|^2_P \Big\}
- V( \textbf{z}_k^* )\\
&~ \leq \mathds{E} \Big\{ \| \textbf{z}_{k+N_p|k}^* \|^2_Q
+ \|  \textbf{z}_{k+N_p|k}^* \|^2_{K^T R K}
+ \| \textbf{z}_{k+N_p|k}^* \|^2_{A_{cl}^T P A_{cl}} \\
&~~~~ + \| B_{\varphi} \varphi_{k+N_p|k} \|^2_{P}
+  \| G \textbf{w}_{k+N_p|k} \|^2_{P}
- \| \textbf{z}_{k|k}^* \|^2_Q
- \| \textbf{u}_{k|k}^* \|^2_R \\
&~~~~ - \| \textbf{z}_{k+N_p|k}^* \|^2_P \Big\}\\
&~ = \mathds{E} \Big\{-\| \textbf{z}_{k|k}^* \|^2_Q
- \| \textbf{u}_{k|k}^* \|^2_R
+ \| B_{\varphi} \varphi_{k+N_p|k} \|^2_{P} \\
&~~~~+  \| G \textbf{w}_{k+N_p|k} \|^2_{P} \Big\} \\
&~ \leq - \| \textbf{z}_{k|k}^* \|^2_Q
+ \| B_{\varphi} \varphi_{k+N_p|k} \|^2_{P}
+ \mathds{E} \Big\{ \| G \textbf{w}_{k+N_p|k} \|^2_{P} \Big\}
\end{split}
\label{eq:ISS}
\end{equation}
where $\textbf{s}_{k|k}^* = \textbf{z}_{k|k}^*$, and $A_{cl}^T P A_{cl} + K^T R K + Q = P$ since $P$ is the solution of~\eqref{eq:riccati equation}. Therefore, $V( \textbf{z}_k^* )$ is the ISS-Lyapunov function and the closed-loop system is input-to-state stable.
\end{proof}
Moreover, summing~\eqref{eq:ISS} over $k=0, 1, \dots$ leads to
\begin{equation}
\lim_{n \rightarrow \infty}\frac{1}{n} \sum_{k=0}^{n} \mathds{E} (  \| \textbf{z}_{k} \|^2_Q
+ \| \textbf{u}_{k} \|^2_R )
\leq
L_{ss}
\label{eq:mean square stability}
\end{equation}
where $L_{ss}=\lim_{n \rightarrow \infty} \sum_{k=0}^{n} \mathds{E} ( \| B_{\varphi} \varphi_k \|^2_{P}
+ \| G \textbf{w}_k \|^2_{P})/n$
by using discrete-time version of Dynkin's Formula~\cite{kushner1967stochastic}. It is straightforward that the state of the closed-loop system does not converge asymptotically to the origin but remains within a neighborhood of the origin due to the external signal and uncertainty by viewing~\eqref{eq:mean square stability}, which means mean-square stability~\cite{farina2016stochastic, kouvaritakis2010explicit}.

\section{Simulation Results}
\label{sec:results}

\subsection{Simulation Set-up and Koopman Operator-based Vehicle Modeling}
\begin{figure}[t]
\centering
\includegraphics[width=\columnwidth]{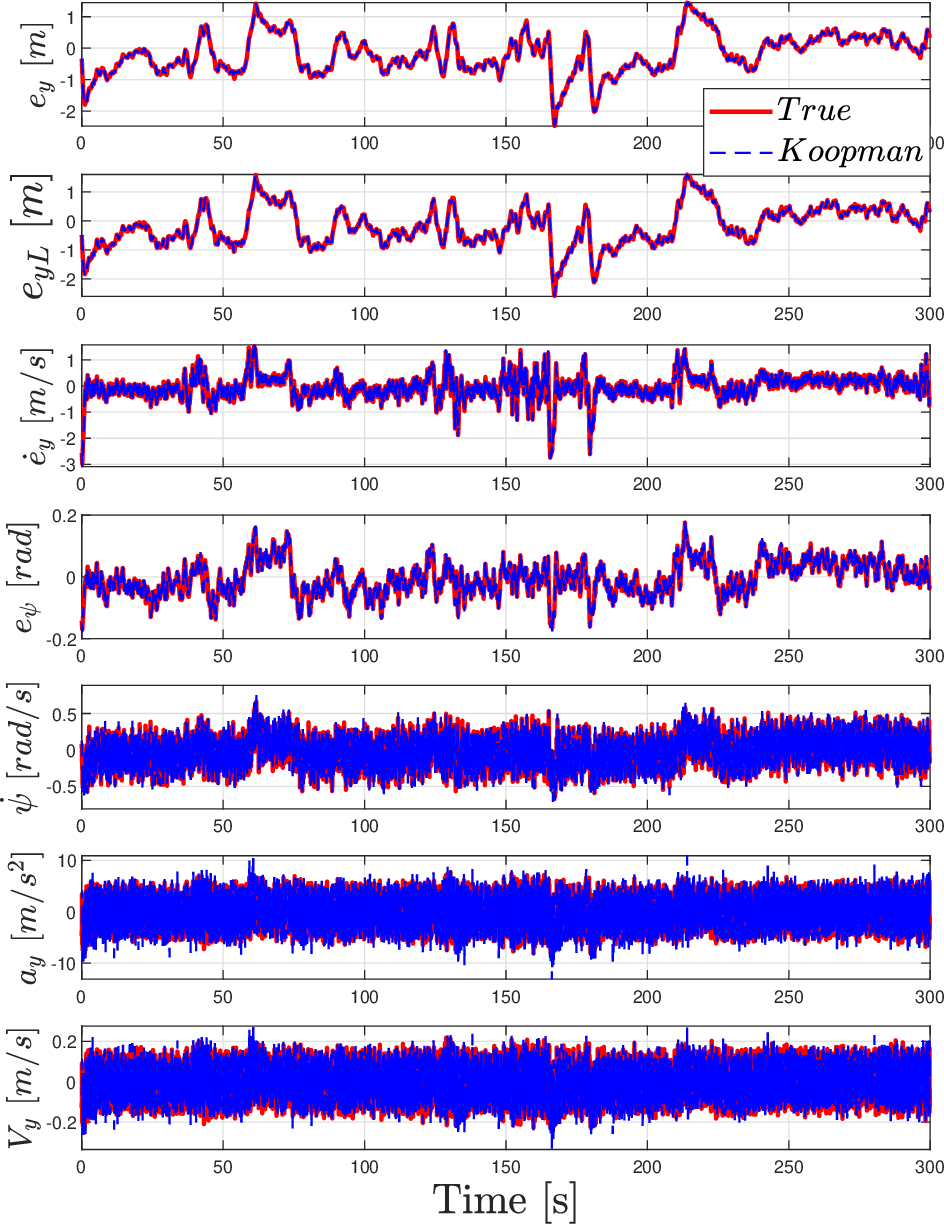}
\caption{ Model fitting accuracy of the Koopman model with validation set}
\label{fig:KoopmanModeling}
\end{figure}

The proposed method was validated using the co-simulation platform with MATLAB/Simulink and CarSim. The vehicle dynamic simulator, CarSim, provides a vehicle model with 27 degrees of freedom for representing the highly nonlinear vehicle dynamics allowing for testing of the realistic motion of a vehicle.
We used various roads provided by CarSim to obtain the training dataset for computing the Koopman operator with a sample time of $0.01s$. Some of the system states are related to the given road lane, i.e., $e_y$, $e_{yL}$, $\dot{e}_y$, and $e_{\psi}$; thus a path-follow controller stabilizing the vehicle lateral motion is needed to obtain the dataset.
Moreover, random signals are added to the input to sufficiently excite the nonlinear vehicle dynamics. For more details, please refer to our previous work~\cite{kim2022data, kim2023koopman}. Then, the dataset matrices~\eqref{eq:data collection} is obtained with $M = 1.22 \times 10^5$. We chose $N=22$ in~\eqref{eq:lifting function} to obtain the lifted state.
In addition, it is reported that a thin plate spline radial basis function is an effective lifting function in autonomous vehicle modeling compared to the other basis functions~\cite{kim2023koopman}. Thus, the nonlinear lifting functions $\phi_i$ are selected as the thin plate spline radial basis functions, i.e., $\phi_i(\textbf{x})=\|\textbf{x}-\textbf{c}_l \|_2^2 \cdot \text{log}\|\textbf{x}-\textbf{c}_l\|_2$ where $\textbf{c}_l$ is randomly selected with a uniform distribution in a certain range~\cite{korda2018linear}. The number of thin plate spline radial basis functions is set to $15$ in~\eqref{eq:lifted state}.

Based on the obtained training dataset, we approximate the Koopman operator in the finite-dimensional space using~\eqref{eq:minimization for AB}. The approximated Koopman operator is tested to validate the modeling accuracy with a validation dataset. The fitting performance is shown in Fig.~\ref{fig:KoopmanModeling}. The red line depicts the true state of the vehicle acquired from CarSim, and the blue line illustrates the predicted vehicle state by the Koopman operator-based vehicle model. As shown in Fig.~\ref{fig:KoopmanModeling}, the Koopman-based vehicle model can predict the vehicle state well.

\subsection{Comparative Study}
\begin{figure}[t]
\centering
\subfigure[$V_x$]{
\hspace{-6mm}
\includegraphics[width=1.15\columnwidth]{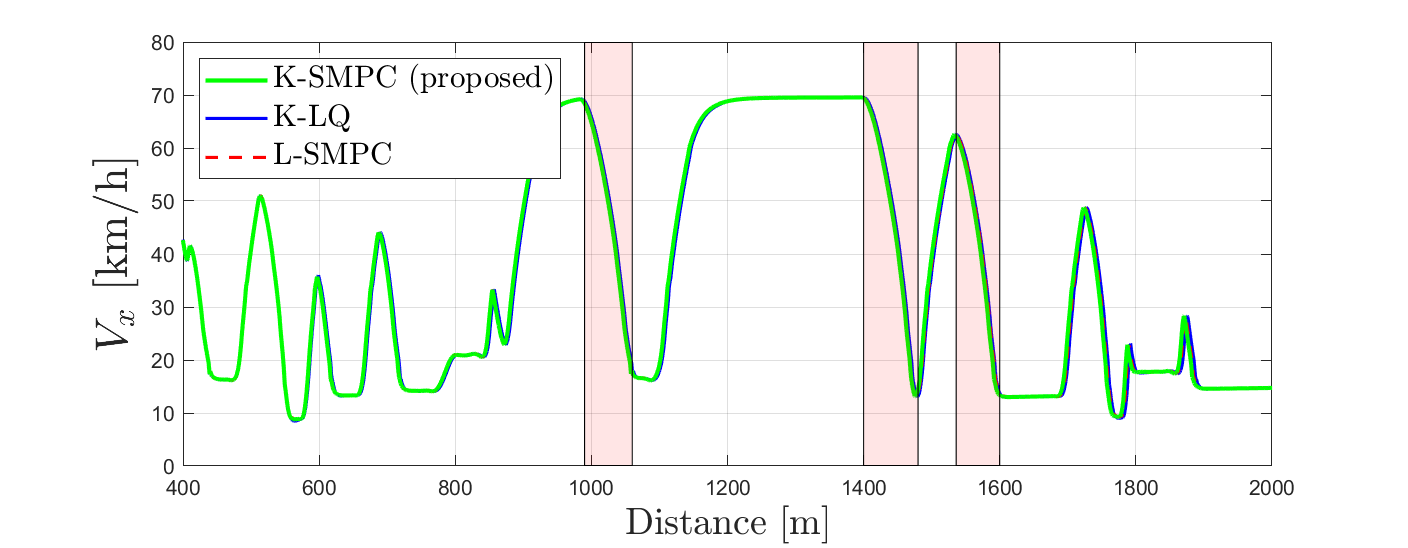}}
\subfigure[$C_0$, $C_1$, $C_2$, and $C_3$]{
\hspace{-6mm}
\includegraphics[width=1.15\columnwidth]{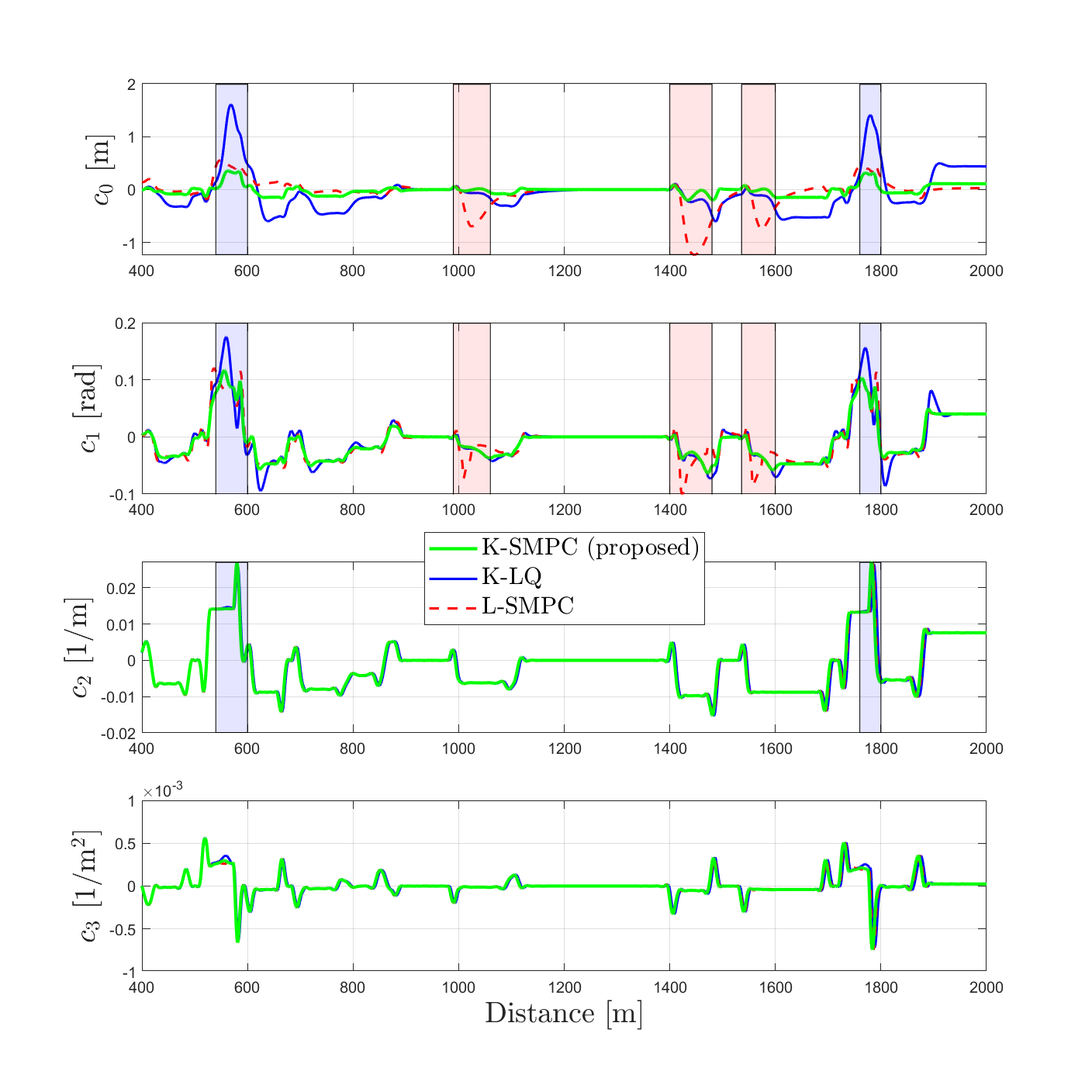}}
\caption{L-SMPC has large tracking error ($C_0$ and $C_1$) in pink-colored section, where vehicle speed is rapidly varying. K-LQ has large tracking error ($C_0$ and $C_1$) in blue-colored section, where road has high curvature: (a) Vehicle longitudinal speed, and (b) road coefficients.}
\label{fig:speed and cam}
\end{figure}
\begin{figure}[t!]
\centering
\includegraphics[width=0.95\columnwidth]{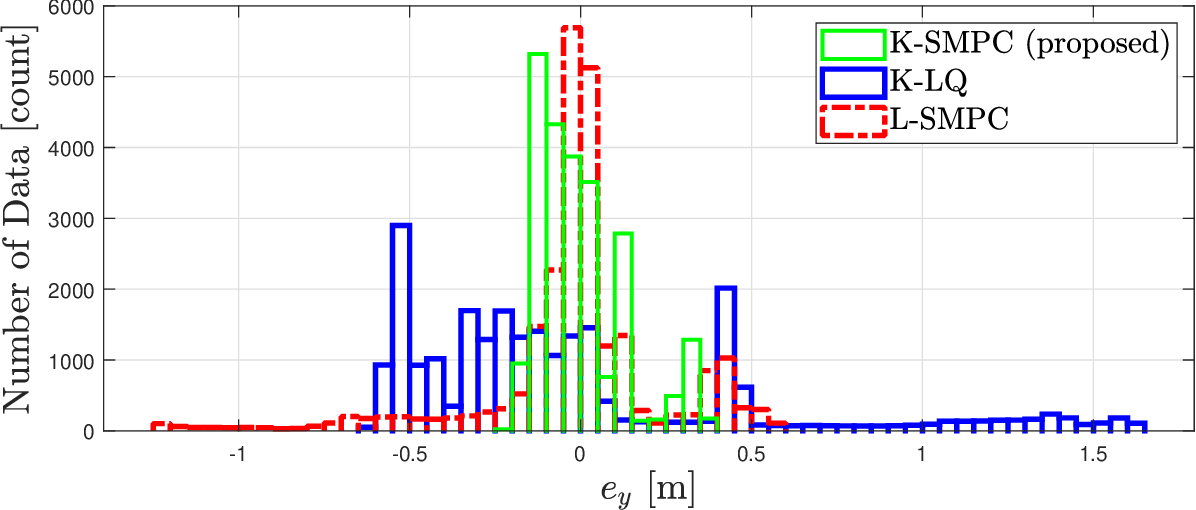}
\caption{$e_y$ histogram}
\label{fig:histogram ey}
\end{figure}
\begin{figure}[t!]
\centering
\subfigure[Lateral position error w.r.t. lane, $e_y$]{
\includegraphics[width=\columnwidth]{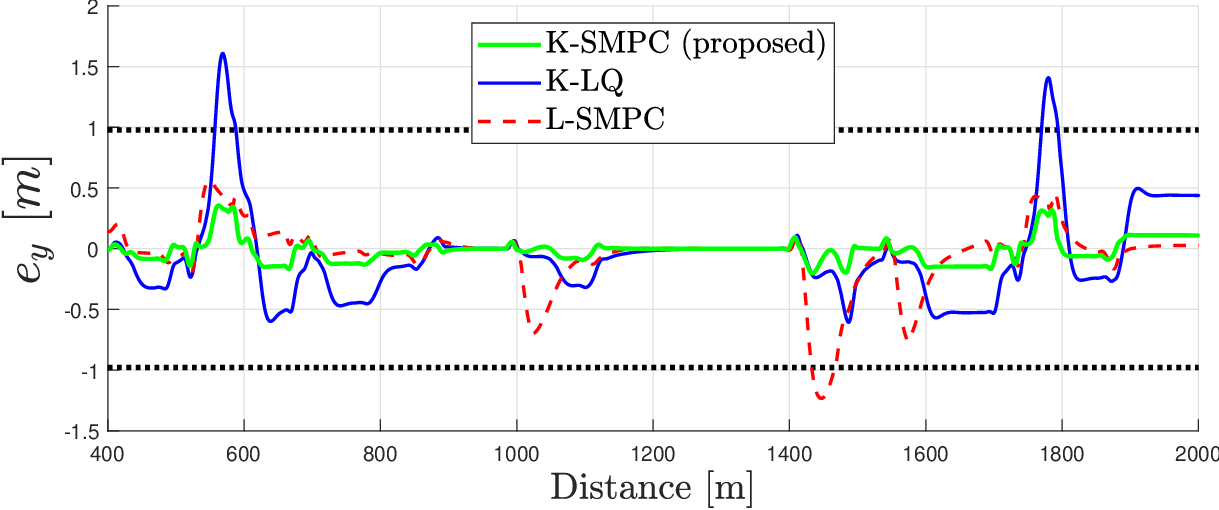}}
\subfigure[Lateral position error on the look-ahead point, $e_{yL}$]{
\includegraphics[width=\columnwidth]{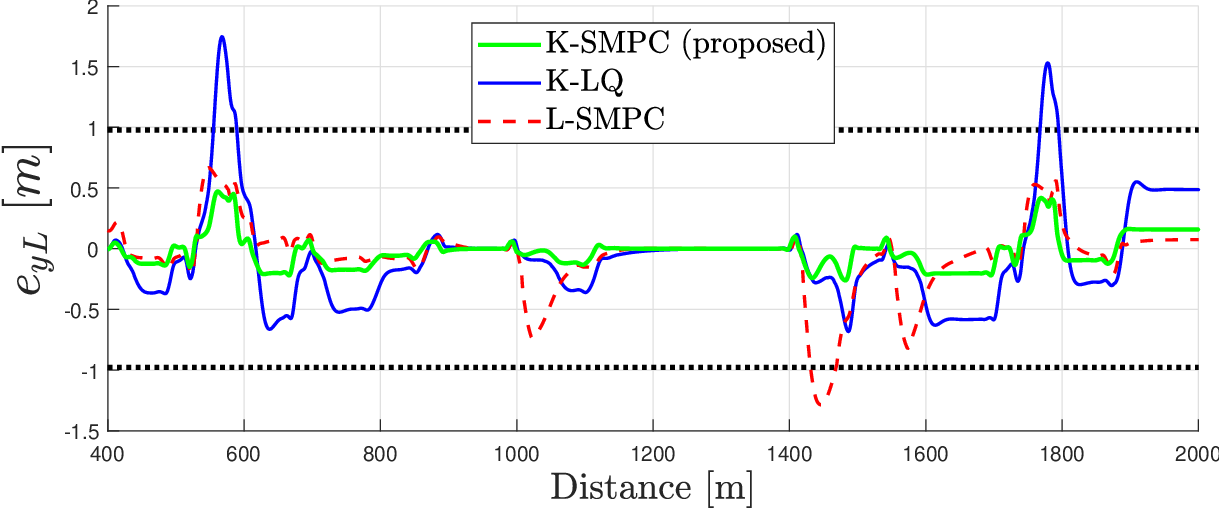}}
\subfigure[Derivative of the lateral position error, $\dot{e}_y$]{
\includegraphics[width=\columnwidth]{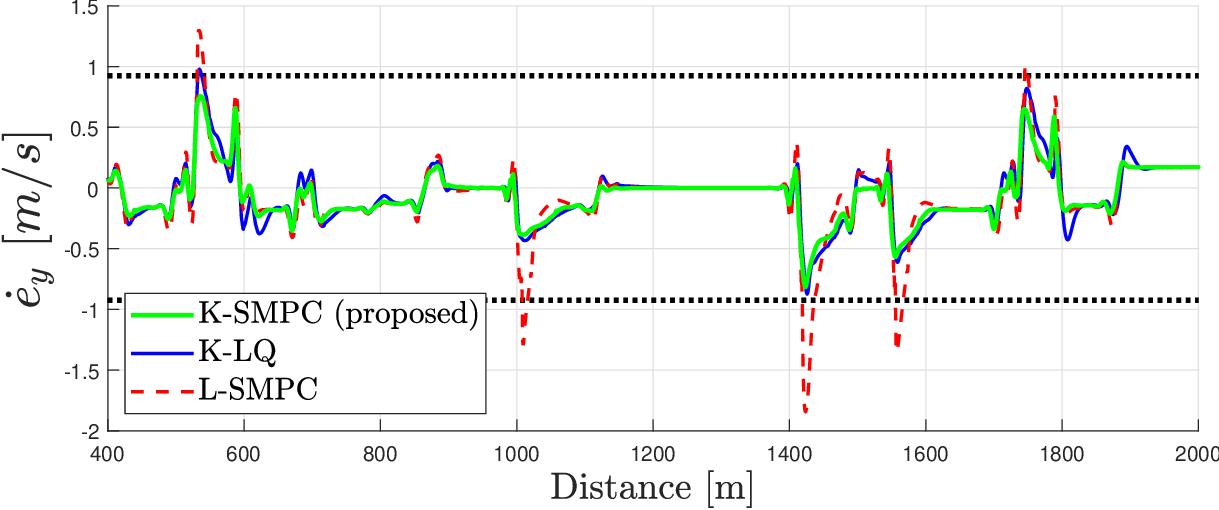}}
\subfigure[Heading angle error w.r.t. lane, $e_{\psi}$]{
\includegraphics[width=\columnwidth]{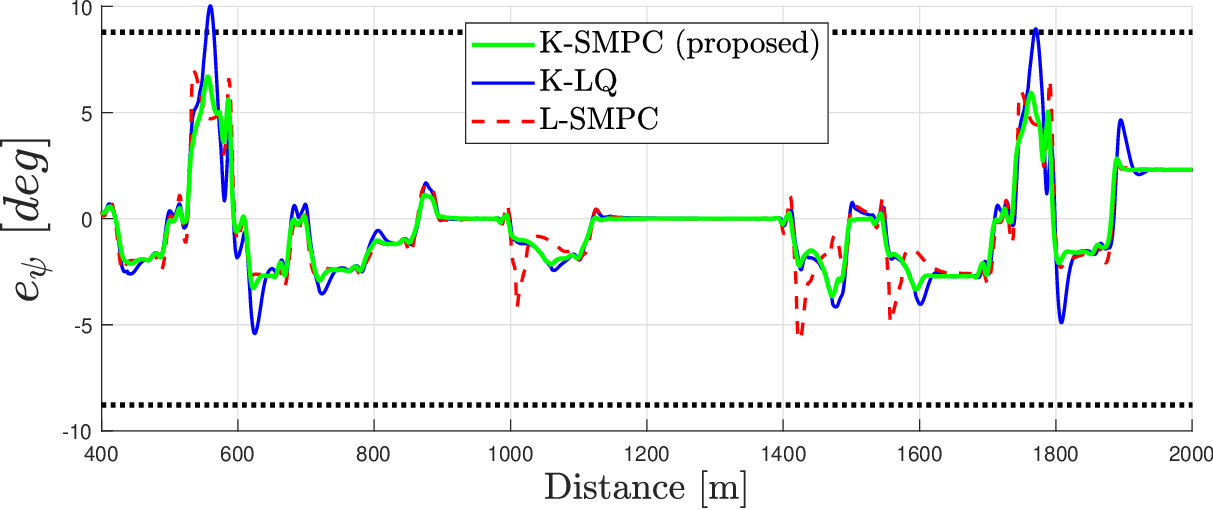}}
\subfigure[Yaw-rate of the vehicle, $\dot{\psi}$]{
\includegraphics[width=\columnwidth]{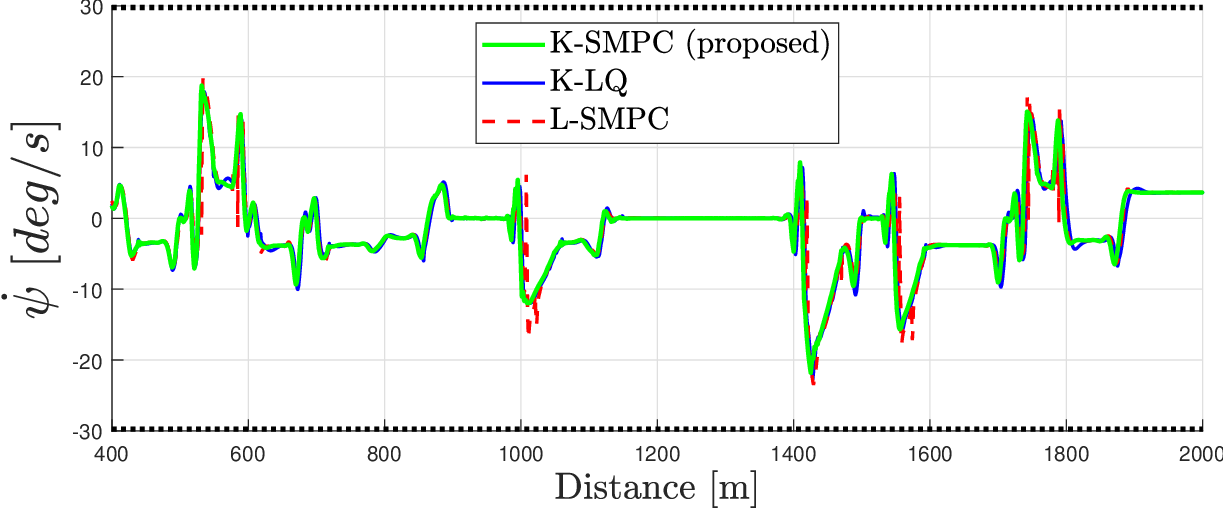}}
\caption{Control results of the system state}
\label{fig:state}
\end{figure}

In order to conduct a comparative study, we adopted two methods, the L-SMPC and the K-LQ. The details of each method are as follows:

\begin{itemize}
\item L-SMPC: The linear vehicle model was adopted as the look-ahead lateral dynamic model from~\cite{lee2016robust, quan2022linear} given by
    \begin{equation}
    \dot{\textbf{x}}_v = A_v{\textbf{x}}_v + B_v \textbf{u}_v + B_{v \varphi} \varphi_v,
    \end{equation}
    where
    \begin{gather*}
    A_v=\begin{bmatrix} 0&1&0&-L\\
                     0&a_{22}&a_{23}&a'_{24}\\
                     0&0&0&-1\\
                     0&a'_{42}&a_{43}&a_{44}\end{bmatrix},~
                     B_v=
                     \begin{bmatrix}
                     0\\b'_{21} \\0\\b_{41}
                     \end{bmatrix},\\
                     B_{v \varphi}=
                     \begin{bmatrix}
                     L&V_{x}\\V_{x}&0\\1&0\\0&0
                     \end{bmatrix},
    ~\textbf{u}_v = \delta ,~ \varphi_v =\begin{bmatrix} {\dot{\psi}}_{des} \\ e_{\psi L}-e_{\psi}\end{bmatrix},
    \end{gather*}
    %
    %
    with
    \begin{gather*}
    a_{22}=-\frac{2C_{\alpha f}+2C_{\alpha r}}{mV_{x}},~  a_{23}=-a_{22}V_{x},\\
    a_{24}=-1-\frac{2C_{\alpha f}l_{f}-2C_{\alpha r}l_{r}}{m{V_{x}}^2},~ a'_{24}=(a_{24}-1)V_{x},\\
    a_{42}=-\frac{2C_{\alpha f}l_{f}-2C_{\alpha r}l_{r}}{I_{z}},~ a'_{42}=a_{42}/V_{x},\\
    a_{43}=-a_{42},~ a_{44}=-\frac{2C_{\alpha f}{l_{f}}^{2}+2C_{\alpha r}{l_{r}}^{2}}{I_{z}V_{x}},\\
    b_{21}=\frac{2C_{\alpha f}}{mV_{x}},~ b'_{21}=b_{21}V_{x},~ b_{41}=\frac{2C_{\alpha f}l_{f}}{I_{z}}.
    \end{gather*}
    Then, the linear vehicle model was discretized with the zero-order-holder method. We designed the SMPC to be similar to~\eqref{eq:smpc problem} except for the system model~\eqref{eq:smpc problem_system}. The linear model-based SMPC for the LKS was successfully studied in~\cite{zhang2021stochastic}. However, the linear model is not appropriate since the cornering stiffness is no longer linear with respect to the tire slip angle when the road curvature is high and vehicle speed rapidly changes~\cite{zhang2021stochastic, rajamani2011vehicle}. Therefore, we can confirm the effectiveness of the proposed method in dynamic lane-keeping scenarios by comparing the result of the L-SMPC.

\item K-LQ~\cite{kim2023koopman}: The K-LQ method uses the Koopman operator-based vehicle model, the same as~\eqref{eq:lift_sys}. However, the linear quadratic regulator was adopted to control the system, i.e., the only difference with the proposed method is the control scheme. From~\cite{kim2023koopman}, the road information, i.e., $\varphi$, was not considered in the controller design. Thus, we can evaluate the tracking performance of the proposed scheme on high-curvature roads by comparing the performance of the K-LQ.
\end{itemize}

For fair comparison with the proposed method, we used the same weighting matrix on the system state and control input in the design of the controller of each method. Moreover, this paper designs the longitudinal controller to control the vehicle speed with respect to the road curvature with a proportional-derivative (PD) controller. The design process of the longitudinal controller is out of the scope of this paper; hence, the reader can refer to the authors' work~\cite{quan2022robust} for a detailed description. In Fig.~\ref{fig:speed and cam}~(a), it can be shown that the vehicle longitudinal speed is equal to each method. Thus, the tracking performance of each method only depends on each lateral controller.

\begin{figure}[t]
\hspace{-8mm}
\centering
\includegraphics[width=1.1\columnwidth]{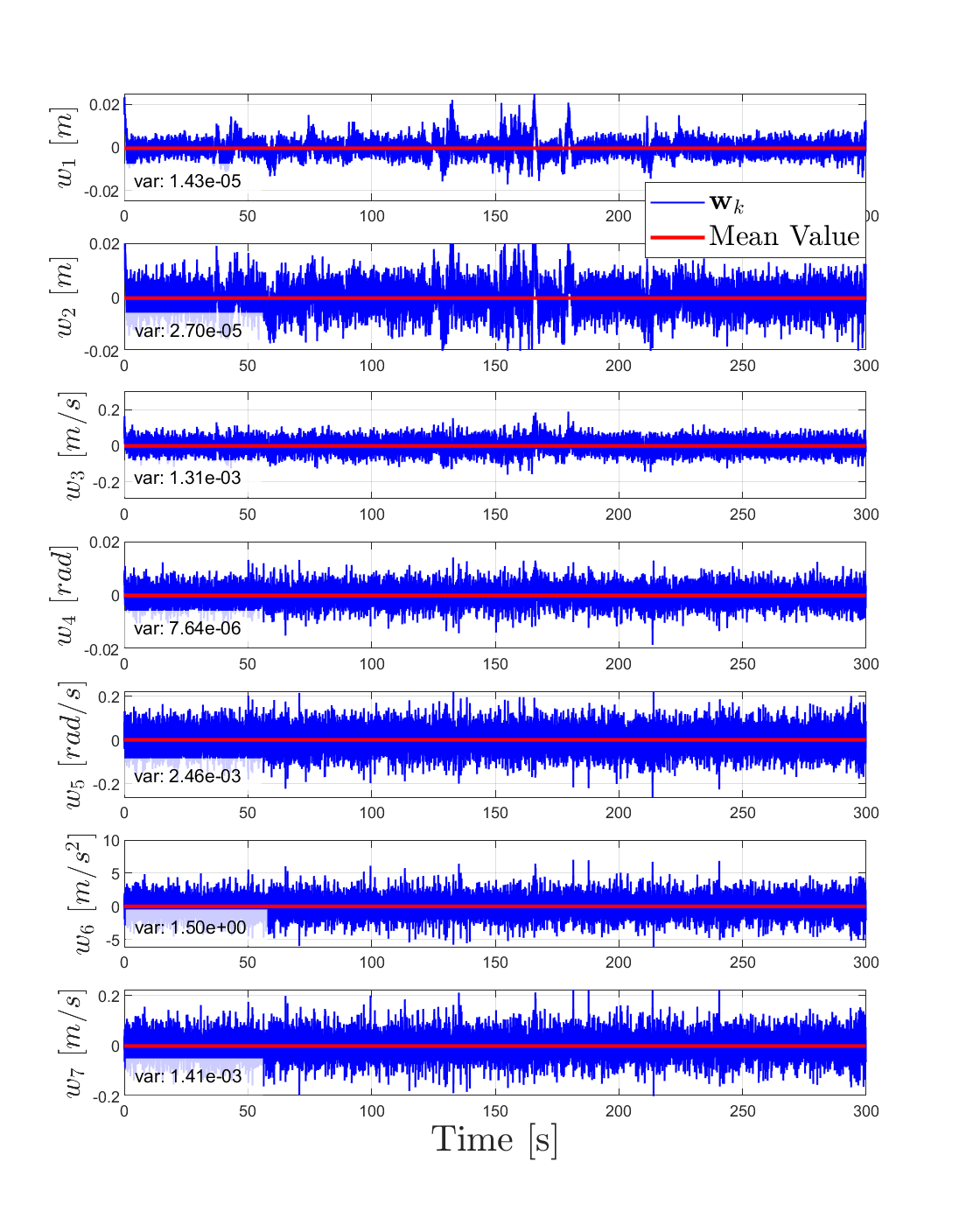}
\vspace{-10mm}
\caption{Uncertainty $\textbf{w}_k$ for each system state}
\label{fig:uncertainty}
\end{figure}
\begin{figure}[t]
\centering
\subfigure[Steering wheel angle]{
\includegraphics[width=\columnwidth]{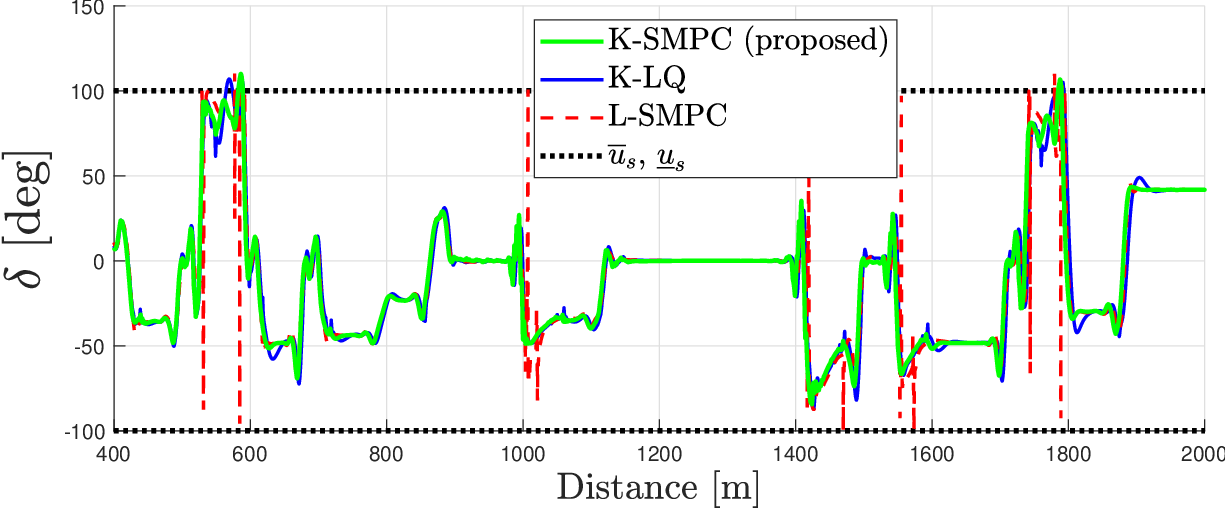}}
\subfigure[Constraint softening slack variables, $\overline{\sigma}$ and $\underline{\sigma}$]{
\includegraphics[width=\columnwidth]{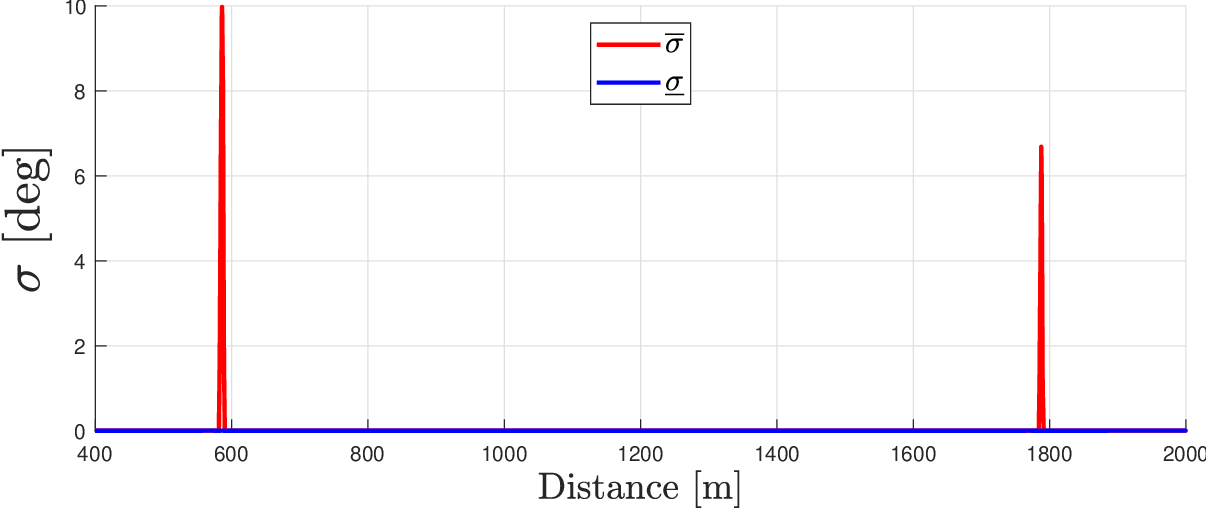}}
\caption{Results of control input and slack variables}
\label{fig:control input}
\end{figure}
\begin{figure}[t]
\centering
\includegraphics[width=\columnwidth]{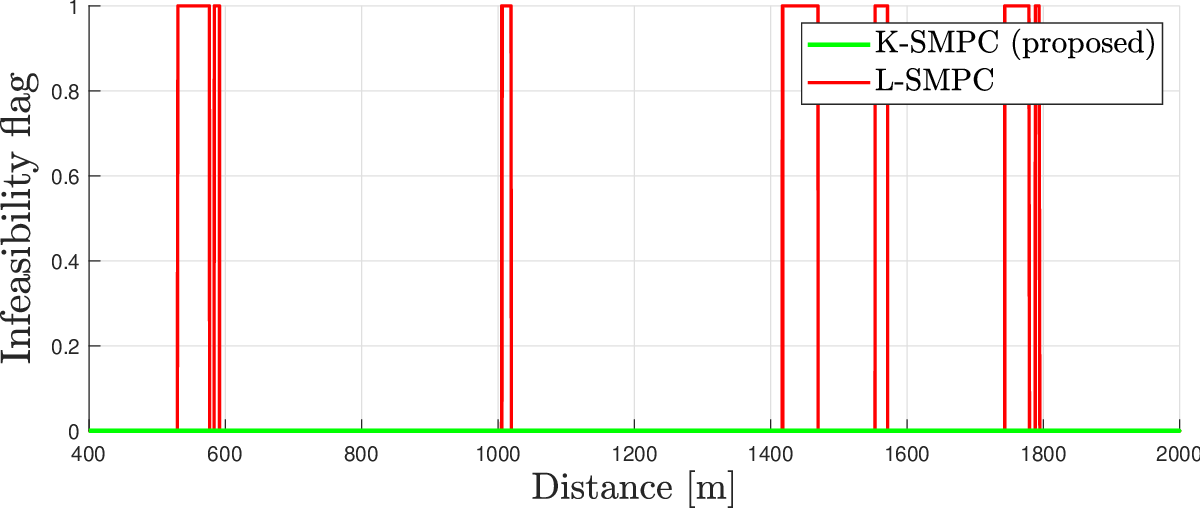}
\caption{SMPC infeasibility: 0 for feasible, 1 for infeasible}
\label{fig:mpc infeasible}
\end{figure}

We use the race-track road provided by CarSim to validate the tracking performance of each method. The road lane coefficients, i.e., $C_0$, $C_1$, $C_2$, and $C_3$ in~\eqref{eq:road model}, are illustrated in Fig.~\ref{fig:speed and cam}~(b). The blue line represents the result of the K-LQ, the red line is the L-SMPC, and the green line depicts the result of the proposed method.
It should be noted that $C_2$ is the curvature of the lane, representing the road shape. Hence, each control method was conducted on the same path.
As mentioned in~\ref{subsec:road lane}, $C_0$ denotes the lateral offset error, and $C_1$ denotes the heading angle error. We can see that the proposed method has a lower lateral position error and a lower heading angle error, i.e., $C_0$ and $C_1$, compared to other methods in Fig.~\ref{fig:speed and cam}~(b).
In particular, the proposed controller had a lower error on the lane with a high curvature, as shown by the blue section in Fig.~\ref{fig:speed and cam}~(b). However, the L-SMPC had a large error in the pink section in Fig.~\ref{fig:speed and cam}~(b) compared to other methods because the linear vehicle model is no longer accurate with rapid varying of vehicle speed~\cite{rajamani2011vehicle}, as illustrated by the pink section in Fig.~\ref{fig:speed and cam}~(a).
On the other hand, the K-LQ can track the given lane even with rapid speed changes because the Koopman operator-based vehicle model can represent the highly nonlinear vehicle dynamics.
However, the K-LQ has a larger error on roads with a high curvature, as depicted by the blue region in Fig.~\ref{fig:speed and cam}~(b). This is because the K-LQ does not consider the future system state, while the K-SMPC and the L-SMPC predict the future state with the curved road information from $\varphi$ in the optimization problem.
As mentioned before, $C_0$ denotes the lateral offset error, which is equal to the system state $e_y$. Thus, we observe that the proposed method has the lowest lateral error by the statistical way, i.e., histogram, as is shown in Fig.~\ref{fig:histogram ey}.

The results of the controlled system state are observed in Fig.~\ref{fig:state}. As defined by the state of the system in~\eqref{eq:state}, some states represent the path-tracking performance. Specifically, $e_y$ and $e_{yL}$ are the lateral position errors at CG and the look-ahead distance, respectively. In addition, $\dot{e}_y$ is the lateral speed tracking error, and $e_{\psi}$ is the heading angle tracking error.
In Fig.~\ref{fig:state}, the blue line represents the K-LQ, the red line represents the L-SMPC result, the green line represents the proposed method, and the black line represents the tightened constraints of each state. We set the constraints as $|e_y| \leq 1~m$, $|e_{yL}| \leq 1~m$, $|\dot{e}_y| \leq 0.95~m/s$, $|e_{psi}| \leq 10~deg$, and $|\dot{\psi}| \leq 30~deg/s$, respectively.
Besides, we set the covariance matrix $\Sigma_w$ from the data. The uncertainty $\textbf{w}_k$ is computed, as shown in Fig.~\ref{fig:uncertainty}. We calculate the variance for each state, then we use the variance for design the covariance matrix. The covariance matrix is defined as
$ \Sigma_w =
\text{diag}\begin{bmatrix} \sigma_1 & \sigma_2 & \sigma_3 & \sigma_4 & \sigma_5 & \sigma_6 & \sigma_7 & \textbf{0}_{1 \times (N-7) }\end{bmatrix}
$
where $\sigma_1=1.43e-5$, $\sigma_2=2.70e-5$, $\sigma_3=1.31e-3$, $\sigma_4=7.64e-6$, $\sigma_5=2.46e-3$, $\sigma_6=1.5$, and $\sigma_7=1.41e-3$.
%
%
%
%
\begin{table}[b]
\caption{Comparison of controller performance on validation road}
\begin{center}
\small
\begin{tabular}{cccccc}
\hline \hline
\multicolumn{6}{c}{Root Mean Squared Error}\\
\hline \hline
\vspace{1mm}
State & $e_y$  & $e_{yL}$ & $\dot{e}_y$  & $e_{\psi}$ & $\dot{\psi}$ \vspace{1mm} \\
\hline
\vspace{1mm}
K-LQ   & 0.517 & 0.571 & 0.252 & 0.055  & 0.088   \\
%
L-SMPC   & 0.266 & 0.299 & 0.271 & \textbf{0.043}  & 0.089 \\
K-SMPC & \textbf{0.130} & \textbf{0.180} & \textbf{0.220} & 0.045 & \textbf{0.088} \\
\hline
\end{tabular}
\label{table:RMSE_chap5}
\end{center}
\end{table}
\begin{table}[b]
\caption{Comparison of controller performance on validation road}
\begin{center}
\small
\begin{tabular}{cccccc}
\hline \hline
\multicolumn{6}{c}{Max Error}\\
\hline \hline
\vspace{1mm}
State & $e_y$  & $e_{yL}$ & $\dot{e}_y$  & $e_{\psi}$ & $\dot{\psi}$ \vspace{1mm} \\
\hline
\vspace{1mm}
K-LQ   & 1.608 & 1.746 & 0.978 & 0.175  & 0.396   \\
%
L-SMPC   & 1.233 & 1.285 & 1.845 & 0.121  & 0.410 \\
K-SMPC & \textbf{0.356} & \textbf{0.472} & \textbf{0.817} & \textbf{0.117} & \textbf{0.381} \\
\hline
\end{tabular}
\label{table:ME_chap5}
\end{center}
\end{table}
%
%
%
As shown in Fig.~\ref{fig:state}, it can observed that the proposed method has less error in the lateral position, lateral speed, and heading angle, i.e., $e_y$, $e_{yL}$, $\dot{e}_y$, and $e_{\psi}$. Moreover, the proposed controller satisfies the given constraints of each system state, while other methods violate the constraints in some sections.
%
%
As a result, the proposed method has better tracking performance than the other methods in the LKS application. We can observe the quantitative results in Table.~\ref{table:RMSE_chap5} and Table.~\ref{table:ME_chap5}. It can be confirmed that the lateral position error is dramatically reduced with the proposed method.

In Fig.~\ref{fig:control input}~(a), the control inputs of each method are depicted. The result of the K-LQ is the blue line, the L-SMPC method is the red line, the K-SMPC is the green line, and the first step input constraints are shown as the black line. It can be seen that the L-SMPC has a large oscillation in some ranges because the L-SMPC is infeasible where the given constraints are violated, as shown in Fig.~\ref{fig:mpc infeasible}.
On the other hand, the K-LQ and the K-SMPC method have a smooth control input. In addition, we can observe that the proposed method slightly violates the input constraints at about $600~m$ and $1800~m$ to control the vehicle on the high curvature road. However, note that we consider the soft constraints on the first step of the input. Thus, the slack variables can be observed, as shown in Fig.~\ref{fig:control input}~(b).
\begin{figure}[t!]
\centering
\subfigure[Tire slip angle of the K-LQ result]{
\hspace{-6mm}
\includegraphics[width=1.15\columnwidth]{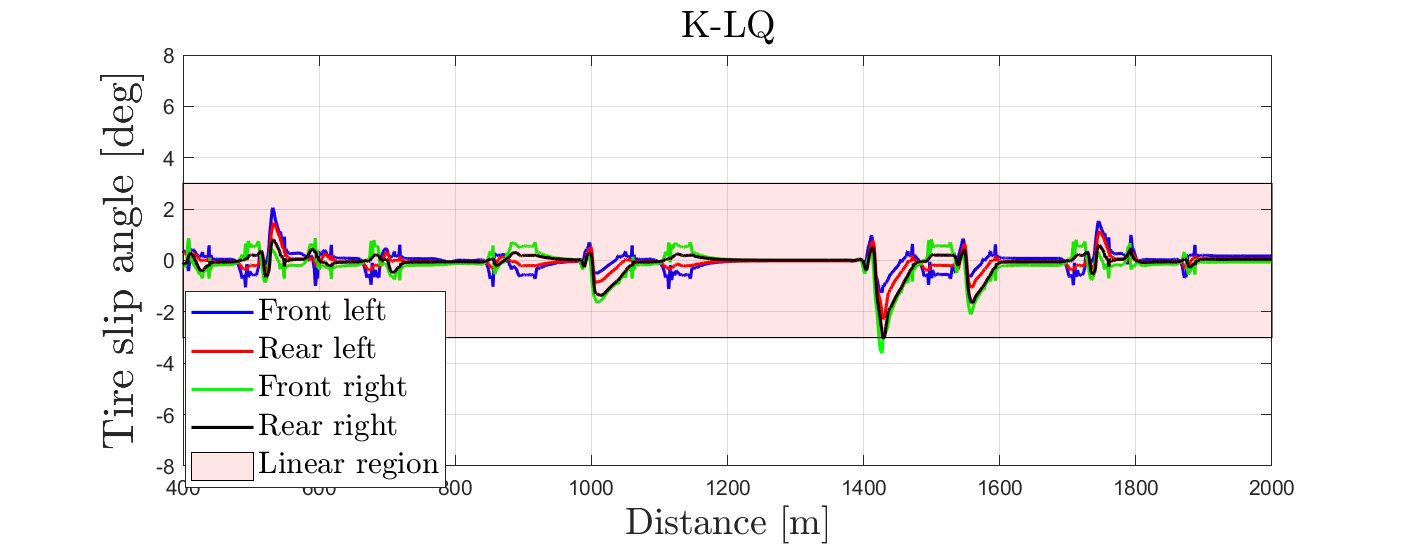}}
\hfill
\subfigure[Tire slip angle of the L-SMPC result]{
\hspace{-6mm}
\includegraphics[width=1.15\columnwidth]{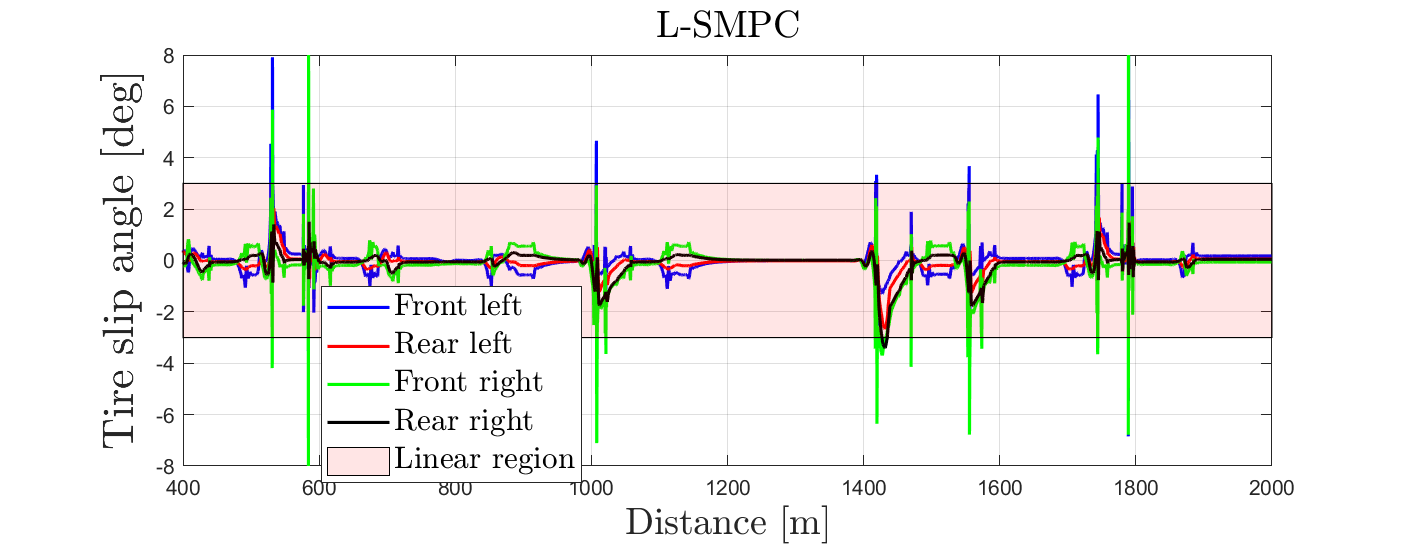}}
\hfill
\subfigure[Tire slip angle of the K-SMPC result]{
\hspace{-6mm}
\includegraphics[width=1.15\columnwidth]{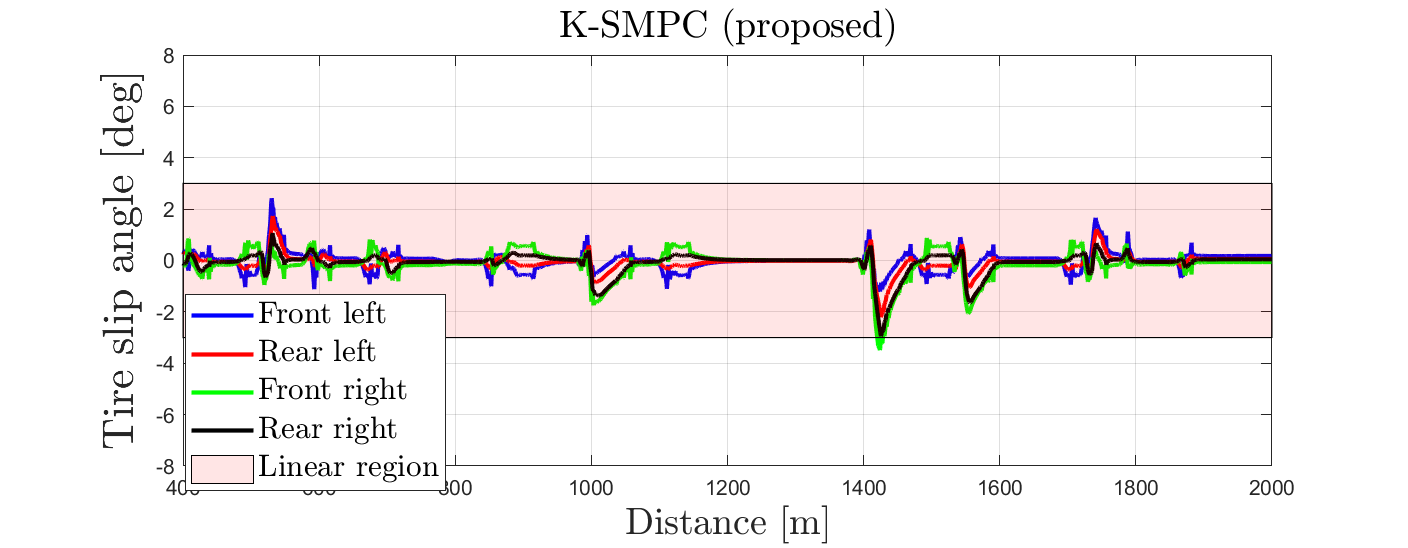}}
\hfill
\caption{Results of the tire slip angle}
\label{fig:slip}
\end{figure}

The results of each tire slip angle are shown in Fig.~\ref{fig:slip} for each method. The pink section of Fig.~\ref{fig:slip} represents the linear relationship between the lateral tire force and the tire slip angle~\cite{rajamani2011vehicle, kim2023koopman}, i.e., the cornering stiffness is a linear function of the tire slip angle in~\eqref{eq:tire force}. In this paper, the linear region is selected within $\pm 3~deg$ of the tire slip angle because the lateral tire force and the tire slip angle can be in a linear relationship provided by CarSim data and references~\cite{rajamani2011vehicle, kim2023koopman}.
The Koopman-based vehicle model (i.e., K-LQ and K-SMPC) maintains the tire slip in the linear region. Thus, it can be seen that the Koopman-based model captures the vehicle's nonlinear behavior and effectively controls the vehicle under dynamic situations. On the other hand, the L-SMPC method leaves the linear region so that the linear vehicle model is no longer valid.

\section{Conclusion}
\label{sec:conclusion}

In this paper, we proposed the K-SMPC for the enhanced LKS of autonomous vehicles. The EDMD method was used to approximate the Koopman operator in a finite-dimensional space for practical implementation.
We considered the modeling error of the approximated Koopman operator in the EDMD method. The modeling error was handled as a probabilistic signal. Then, we designed K-SMPC to tackle the modeling error.
The recursive feasibility of the proposed method was guaranteed with the explicit first-step state constraint by computing the robust control invariant set.
For the simulation, a high-fidelity vehicle simulator, CarSim, was used to validate the effectiveness of the K-SMPC. We conducted a comparative study between K-LQ and L-SMPC. From the results, it was confirmed that the proposed method outperforms other methods with respect to the tracking performance. Furthermore, we observed that the proposed method satisfies the given constraints and is recursively feasible.
As future work, a comparative study will be conducted with the Koopman-based RMPC to evaluate the conservativeness quantitatively. Future research may also include a real-car experiment with the proposed method.


\bibliographystyle{IEEEtran}
\bibliography{TIV_koopman_SMPC}

\begin{IEEEbiography}
[{\includegraphics[width=1in,height=1.25in,clip,keepaspectratio]{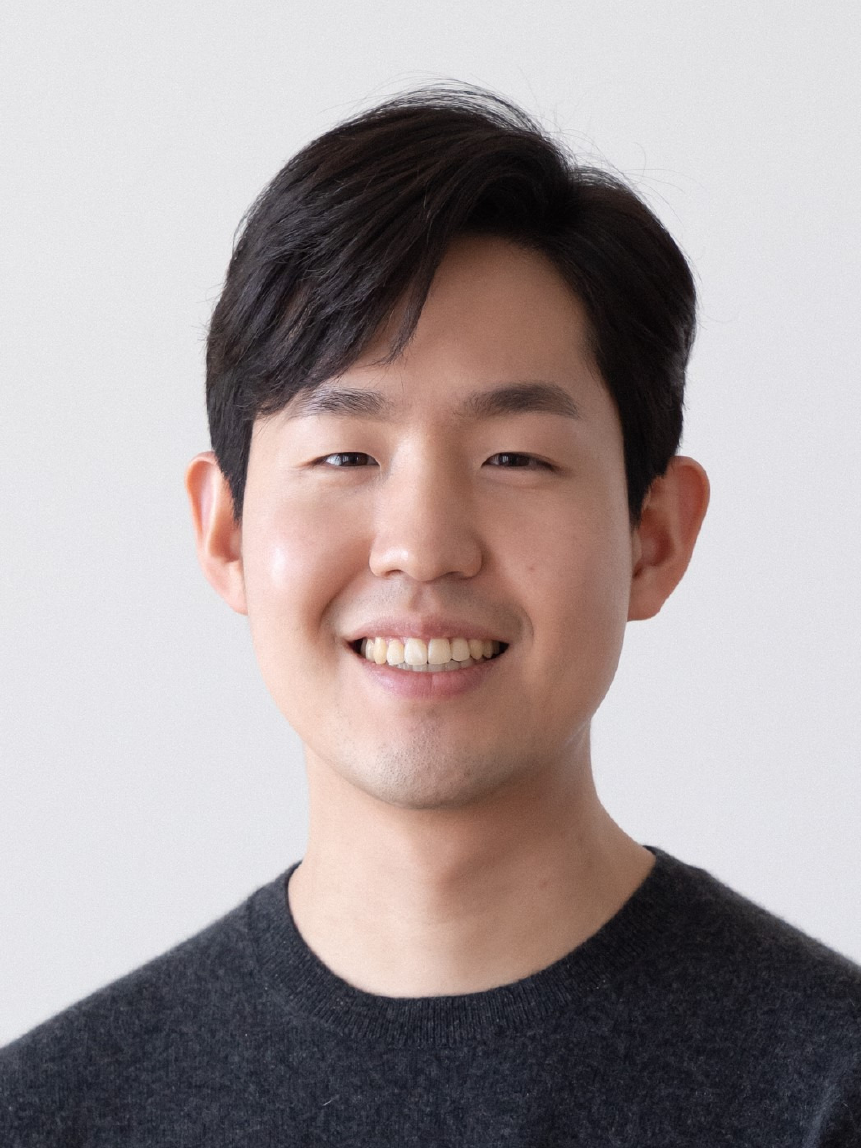}}]
{Jin Sung Kim}
(Graduate Student Member, IEEE) received his B.S. degree in electronic engineering from Kookmin University, Seoul, South Korea, in 2014, and M.S. degree in electrical engineering from Hanyang University, Seoul, in 2019, where he is currently pursuing his Ph.D. degree in electrical engineering. His main research interests include data-driven control, automated vehicles, optimal control, robust control, and artificial intelligence. Jin Sung Kim received the Outstanding Student Paper Award at International Conference on Control, Automation
and Systems in 2022. He is a member of the \textsc{IEEE Intelligent Transportation Systems Society (ITSS)}, \textsc{IEEE Control System Society (CSS)}, the  \textsc{Society of Automotive Engineers (SAE)}, the  \textsc{Korean Society of Automotive Engineers (KSAE)}, and the  \textsc{Institute of Control, Robotics and Systems (ICROS)}.
\end{IEEEbiography}

\begin{IEEEbiography}
[{\includegraphics[width=1in,height=1.25in,clip,keepaspectratio]
{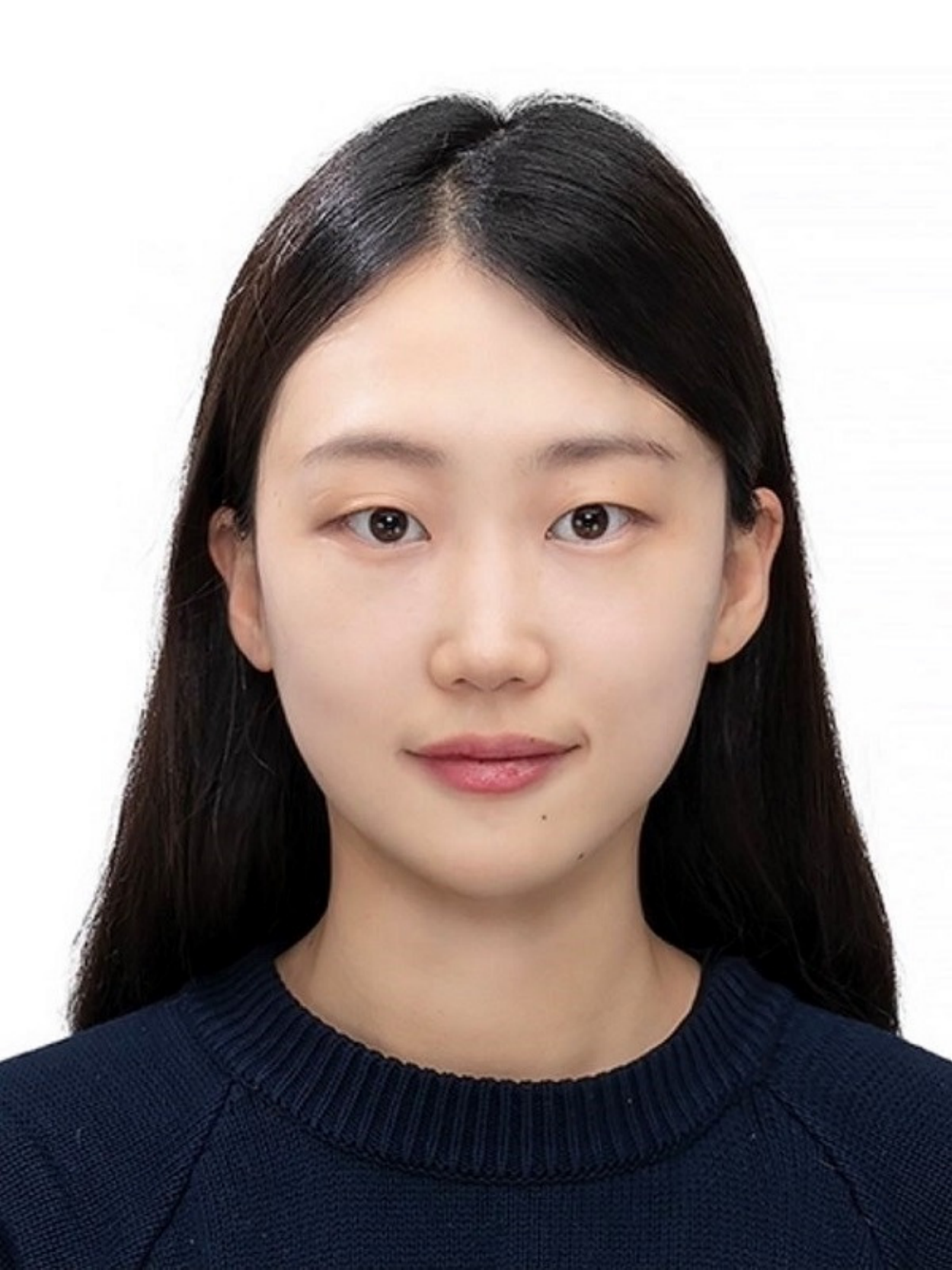}}]
{Ying Shuai Quan}
(Graduate Student Member, IEEE)
received a B.S. degree in control science and engineering in 2017 from the School of Astronautics, Harbin Institute of Technology, Harbin, China. She is working towards her Ph.D. in the Systems and Control Laboratory at the electrical engineering College of Engineering, Hanyang University, Seoul, Korea. Her main research interests include optimal control and learning-based control for autonomous systems. She is a student member of the \textsc{Korean Society of Automotive Engineers} (KSAE) and the \textsc{Institute of Control, Robotics, and Systems} (ICROS).
\end{IEEEbiography}

\begin{IEEEbiography}
[{\includegraphics[width=1in,height=1.25in,clip,keepaspectratio]{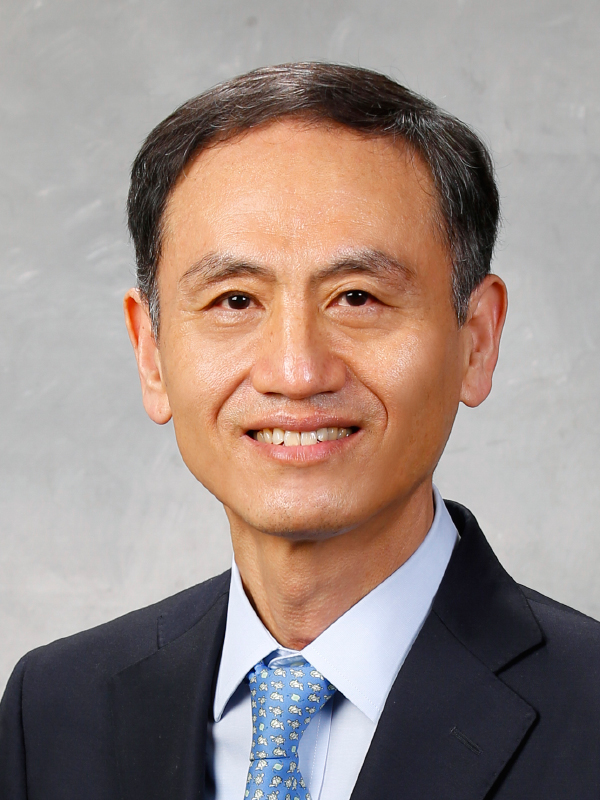}}]
{Chung Choo Chung}
received his B.S. and M.S. degrees in electrical engineering from Seoul National University, Seoul, South Korea, and his Ph.D. degree in electrical and computer engineering from the University of Southern California, Los Angeles, CA, USA, in 1993. From 1994 to 1997, he was with the Samsung Advanced Institute of Technology, South Korea. In 1997, he joined the Faculty of Hanyang University, Seoul, South Korea. Dr. Chung was an Associate Editor for the \textit{Asian Journal of Control} from 2000 to 2002 and an (Founding) Editor for the \textit{International Journal of Control, Automation and Systems} from 2003 to 2005. He served as associate editor for various international conferences, such as the \textsc{IEEE Conference on Decision and Control} (CDC), the American Control Conferences, the IEEE Intelligent Vehicles Symposium, and the Intelligent Transportation Systems Conference. He was a guest editor for a special issue on advanced servo control for emerging data storage systems published by the \textsc{IEEE Transactions on Control System Technologies} (TCST), 2012 and also a guest editor for the \textit{IEEE Intelligent Transportation Systems Magazine}, 2017. He is currently an associate editor for the \textit{IFAC Mechatronics}. He was a program co-chair of ICCAS-SICE 2009, Fukuoka, Japan, an organizing chair for the \textsc{International Conference on Control, Automation and Systems (ICCAS)} 2011, KINTEX, Korea and a program co-chair of the 2015 \textsc{IEEE Intelligent Vehicles Symposium}, COEX, Korea. He was the general chair of ICCAS 2019, Jeju, Korea, and the 2019 President of the Institute of Control, Robotics and Systems (ICROS), Korea. He was also a general chair of CDC 2020, held in Jeju, Korea. He is a member of the National Academy of Engineering of Korea (NAEK).
\end{IEEEbiography}

\end{document}